\documentclass[pdflatex,sn-mathphys-num]{sn-jnl}% Math and Physical Sciences Numbered Reference Style 
\makeatletter
\@twosidefalse
\makeatother

\usepackage{graphicx}%
\usepackage{multirow}%
\usepackage{amsmath,amssymb,amsfonts}%
\usepackage{amsthm}%
\usepackage{mathrsfs}%
\usepackage[title]{appendix}%
\usepackage{xcolor}%
\usepackage{textcomp}%
\usepackage{manyfoot}%
\usepackage{booktabs}%
\usepackage{algorithm}%
\usepackage{algorithmicx}%
\usepackage{algpseudocode}%
\usepackage{listings}%
\usepackage[table]{xcolor}
\usepackage{diagbox}
\usepackage{float}
\usepackage[left]{lineno}
\theoremstyle{thmstyleone}%
\newtheorem{theorem}{Theorem}%  meant for continuous numbers
\newtheorem{proposition}[theorem]{Proposition}% 
\newtheorem{corollary}{Corollary}[theorem]
\theoremstyle{thmstyletwo}%
\newtheorem{lemma}{Lemma}%
\theoremstyle{thmstylethree}%
\newtheorem{definition}{Definition}%

\begin{document}
% \linenumbers

\title[Article Title]{Generative design and validation of therapeutic peptides for glioblastoma based on a potential target ATP5A}

% \author[]{\textbf{Anonymous Authors}}
% \affil[]{Affiliations withheld for double-blind review}

\author[1]{\fnm{Hao} \sur{Qian}}
% \email{qhonearth@sjtu.edu.cn}
\equalcont{These authors contributed equally to this work.}
\author[2]{\fnm{Pu} \sur{You}}
% \email{youpu@quietdbio.com}
\equalcont{These authors contributed equally to this work.}

\author[1]{\fnm{Lin} \sur{Zeng}}
% \email{zenglin10@sjtu.edu.cn}
\author[1]{\fnm{Jingyuan} \sur{Zhou}}
% \email{zjoyuan0930@sjtu.edu.cn}
\author[1]{\fnm{Dengdeng} \sur{Huang}}
% \email{erika\_dd@sjtu.edu.cn}
\author[2]{\fnm{Kaicheng} \sur{Li}}
% \email{kcli@quietdbio.com}
\author*[1]{\fnm{Shikui} \sur{Tu}}\email{tushikui@sjtu.edu.cn}
\author[1]{\fnm{Lei} \sur{Xu}}
% \email{leixu@sjtu.edu.cn}

\affil[1]{ \orgdiv{Centre for Cognitive Machines and Computational Health (CMaCH)}, \orgdiv{School of Computer Science},
  \orgname{Shanghai Jiao Tong University}, \city{Shanghai}, \country{China}}
\affil[2]{\orgdiv{QuietD Biotech}, \city{Shanghai}, \country{China}}

\abstract{
Glioblastoma (GBM) remains the most aggressive tumor, urgently requiring novel therapeutic strategies. Here, we present a dry-to-wet framework combining generative modeling and experimental validation to optimize peptides targeting ATP5A, a potential peptide-binding protein for GBM. Our framework introduces the first lead-conditioned generative model, which focuses exploration on geometrically relevant regions around lead peptides and mitigates the combinatorial complexity of de novo methods. Specifically, we propose POTFlow, a \underline{P}rior and \underline{O}ptimal \underline{T}ransport-based \underline{Flow}-matching model for peptide optimization. POTFlow employs secondary structure information (e.g., helix, sheet, loop) as geometric constraints, which are further refined by optimal transport to produce shorter flow paths. With this design, our method achieves state-of-the-art performance compared with five popular approaches. When applied to GBM, our method generates peptides that selectively inhibit cell viability and significantly prolong survival in a patient-derived xenograft (PDX) model. As the first lead peptide-conditioned flow matching model, POTFlow holds strong potential as a generalizable framework for therapeutic peptide design.

}


\maketitle
\section*{Introduction}
Glioblastoma (GBM) is an aggressive brain malignancy with rapid cell proliferation, diffuse tissue invasion, and a profoundly immunosuppressive microenvironment~\cite{stupp2005radiotherapy,van2022advances}. Despite standard-of-care treatments involving surgical resection followed by radiotherapy and temozolomide chemotherapy, the median survival remains approximately 15 months, with nearly universal recurrence. ATP5A, the alpha subunit and core component of mitochondrial ATP synthase's catalytic site, has emerged as a promising therapeutic target in GBM owing to its abnormal surface localization and critical role in tumor energy metabolism~\cite{xu2016atp5a1,watsonmitochondria}. Targeting ATP5A in GBM presents a compelling therapeutic opportunity, yet its atypical surface exposure poses significant challenges for small-molecule drug design. Peptides offer a promising alternative, as their intermediate molecular size and structural flexibility enable them to access protein surfaces that are typically considered undruggable~\cite{craik2013future, henninot2018pep, wang2022therapeutic}. Moreover, peptides exhibit favorable pharmacological properties, including high target specificity, low immunogenicity, and cost-effective synthesis~\cite{giordano2014neuroactive, fosgerau2015peptide, davda2019immunogenicity}.  

Traditional peptide design approaches, such as alanine scanning~\cite{lefevre1997alanine} and combinatorial screening~\cite{quartararo2020ultra}, are constrained by low throughput and the combinatorial explosion of sequence variants~\cite{fosgerau2015peptide, manning2010stability}. Recently, deep learning-based methods have been increasingly adopted to model protein–peptide interactions~\cite{bhardwaj2016accurate,bryant2022evobind,swanson2022tertiary,cao2022design,bhat2023novo,chen2024pepmlm}. Besides, deep generative models, such as diffusion~\cite{song2020score, song2020denoising,dhariwal2021diffusion} and flow matching models~\cite{lipman2022flow,liu2022flow}, have demonstrated their potential in peptide binder design~\cite{li2024pepflow, lin2024ppflow,kong2024pepglad,wang2024target}. 

% Despite recent successes in designing peptide binders, there are still two limitations. First, current deep generative methods focus mainly on generating peptides from scratch, while in practical drug discovery, peptides are typically derived and optimized from lead sequences, as high sequence similarity often correlates with functional similarity~\cite{orengo1999protein}. Notably, previous studies~\cite{Kaicheng2023Patent,You_qtd} identified a lead peptide sequence targeting ATP5A in GBM for the first time, but its therapeutic efficacy remains suboptimal, underscoring the need for further optimization. Therefore, designing peptides based on lead sequences is a more reasonable direction. 

Despite recent successes in designing peptide binders, there are still two limitations. First, current deep generative methods focus mainly on generating peptides from scratch, while in practical drug discovery, peptides are typically derived and optimized from lead sequences, as high sequence similarity often correlates with functional similarity~\cite{orengo1999protein}. Notably, previous studies~\cite{Kaicheng2023Patent,You_qtd} elucidated the DDIT4L--TOM40--ATP5A pathway in GBM and designed a DDIT4L-derived peptide that binds ATP5A, impairs mitochondrial ATP synthase activity, and suppresses tumor growth in preclinical models. However, this ATP5A-targeting peptide essentially represents a first-generation lead obtained by motif selection, and its potency and therapeutic window remain suboptimal, underscoring the need for systematic lead optimization. Therefore, designing peptides based on lead sequences is a more reasonable and practically relevant direction. Second, the absence of experimental validation in most computational peptide design methods raises concerns about their practical applicability. 

% Last but not least, peptides consist of secondary structures with distinct geometric constraints. However, existing models often ignore these dependencies by sampling residues independently from unconditional priors. This often results in geometrically inconsistent and unnatural peptide conformations.

% These methods typically treat peptides as multimodal data points, representing amino acids as rigid frames on the SE(3) manifold. Each frame is associated with different number of flexible side chains on the torus manifold and a discrete type variable on the simplex manifold. These methods generate novel peptides by starting from random noise and integrating learned reverse-time scores along stochastic or ordinary differential equation (SDE/ODE) trajectories.

To overcome these limitations, we introduce a dry-to-wet peptide design framework for GBM. Unlike de novo methods, our framework explores the local latent space surrounding a lead peptide to identify functionally relevant candidates based on a carefully designed prior. Specifically, we propose a \underline{P}rior and \underline{O}ptimal \underline{T}ransport-based \underline{Flow}-matching (POTFlow) model for peptide optimization. POTFlow initializes a prior distribution using secondary structure types (e.g., helix, sheet, loop) from the lead peptide as geometric constraints. This prior is further refined via optimal transport to shift the sampling distribution closer to the lead peptide, enabling shorter flow paths and improved optimization performance. Besides, by integrating expert system filtering and experimental validation, our framework bridges in silico peptide design with wet-lab evaluation. When applied to GBM, our framework generated peptides that surpassed the lead in inhibiting cell viability in vitro and exhibited tumor cell selectivity. Moreover, it identified a candidate that significantly suppressed tumor growth in a patient-derived xenograft (PDX) model, highlighting its strong potential for therapeutic peptide design.

\textbf{Methodologically, POTFlow extends multimodal flow matching from unconditional de novo generation to a lead-conditioned optimization regime.} To our knowledge, POTFlow is the first lead peptide-conditioned flow-matching frameworks with integrated experimental validation, spanning in silico design, in vitro assays, and in vivo PDX models. 
% Our contributions are summarized as follows:

% \begin{itemize}
%     \item We introduce a secondary-structure–aware prior that initializes residue centroids in a class-specific manner rather than a global Gaussian prior.
%     \item We further apply intra-sample multimodal optimal transport (OT) to couple noisy residues to the lead peptide across four manifolds (positions, orientations, residue types, torsions), yielding shorter and disentangled flow trajectories.
%     \item We embed POTFlow into a dry-to-wet pipeline that starts from a literature-derived lead peptide and ends with PDX-level validation on glioblastoma.
% \end{itemize}

\begin{figure} [th]
    \centering
    % \vspace{-1.5em}
    
    \includegraphics[width=0.98\linewidth]{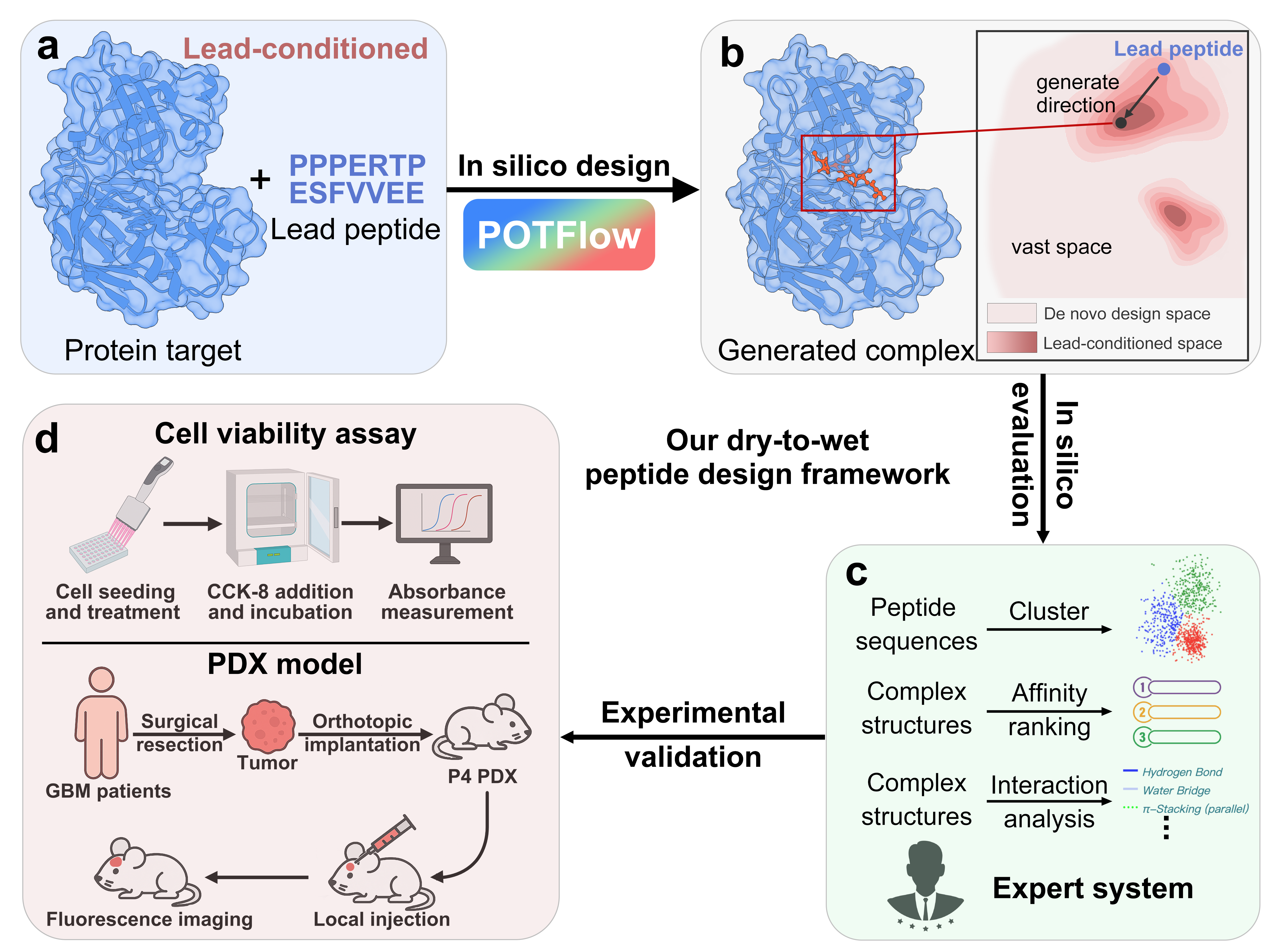}
    \caption{Schematic workflow from in silicon peptide design to experimental validation. \textbf{a)} We start from a lead peptide sequence and 3D structure of its target protein. \textbf{b)} POTFlow efficiently samples peptide–protein complexes within the lead peptide-conditioned space. \textbf{c)} An expert system clusters candidate peptides, ranks their binding affinities, and analyzes intermolecular interactions (e.g., hydrogen bonds, water bridges, $\pi$–stacking). \textbf{d)} Promising candidates are synthesized and tested by cell‐viability assays and patient-derived xenograft (PDX) models. 
}
    \label{fig:workflow}
    % \vspace{-1em}
    
\end{figure}

\section*{Results}

% \begin{enumerate}
%     \item Using existing protein-peptide docking tools, we predict the complex structure based on known binding peptide sequences.
%     \item With POTFlow, we explore the latent space associated with known peptide sequences and their corresponding structures to optimize these sequences.
%     \item These peptides are then clustered on the basis of sequence similarity. Subsequently, domain experts screen these candidate sequences by analyzing the molecule interactions and binding affinities.
%     \item The selected sequences are tested for cellular activity and validated through murine experiments to assess their drug-like properties.
% \end{enumerate}

\subsection*{An overview of the proposed peptide design framework for glioblastoma }

We build a deep generative model jointly with experimental analysis in a dry-to-wet framework to design peptide therapeutics for GBM (Fig.~\ref{fig:workflow}). The core of our framework is POTFlow, the first lead-peptide-driven deep generative model, which fundamentally sets it apart from conventional de novo design methods. De novo design approaches often suffer from inefficiency and low success rates, primarily due to the vast combinatorial search space and the absence of biological context resulting from the use of an unconditional prior. In contrast, POTFlow utilizes the lead peptide to guide and restrict exploration within its local sequence–structure neighborhood (Fig.~\ref{fig:workflow}a–b), thereby focusing the search on geometrically relevant regions and increasing the likelihood of identifying functional candidates. These candidates are subsequently evaluated by an expert system (Fig.~\ref{fig:workflow}c). Top candidates are further synthesized and validated by cell viability assays and patient-derived xenograft (PDX) models (Fig.~\ref{fig:workflow}d). Importantly, our framework bridges computational design with experimental validation which enables efficient translation of in silico design into experimentally validated peptide therapeutics.
% \textcolor{red}{multi-level filtering}

\begin{figure} [tbp]
    \centering
    \includegraphics[width=0.98\linewidth]{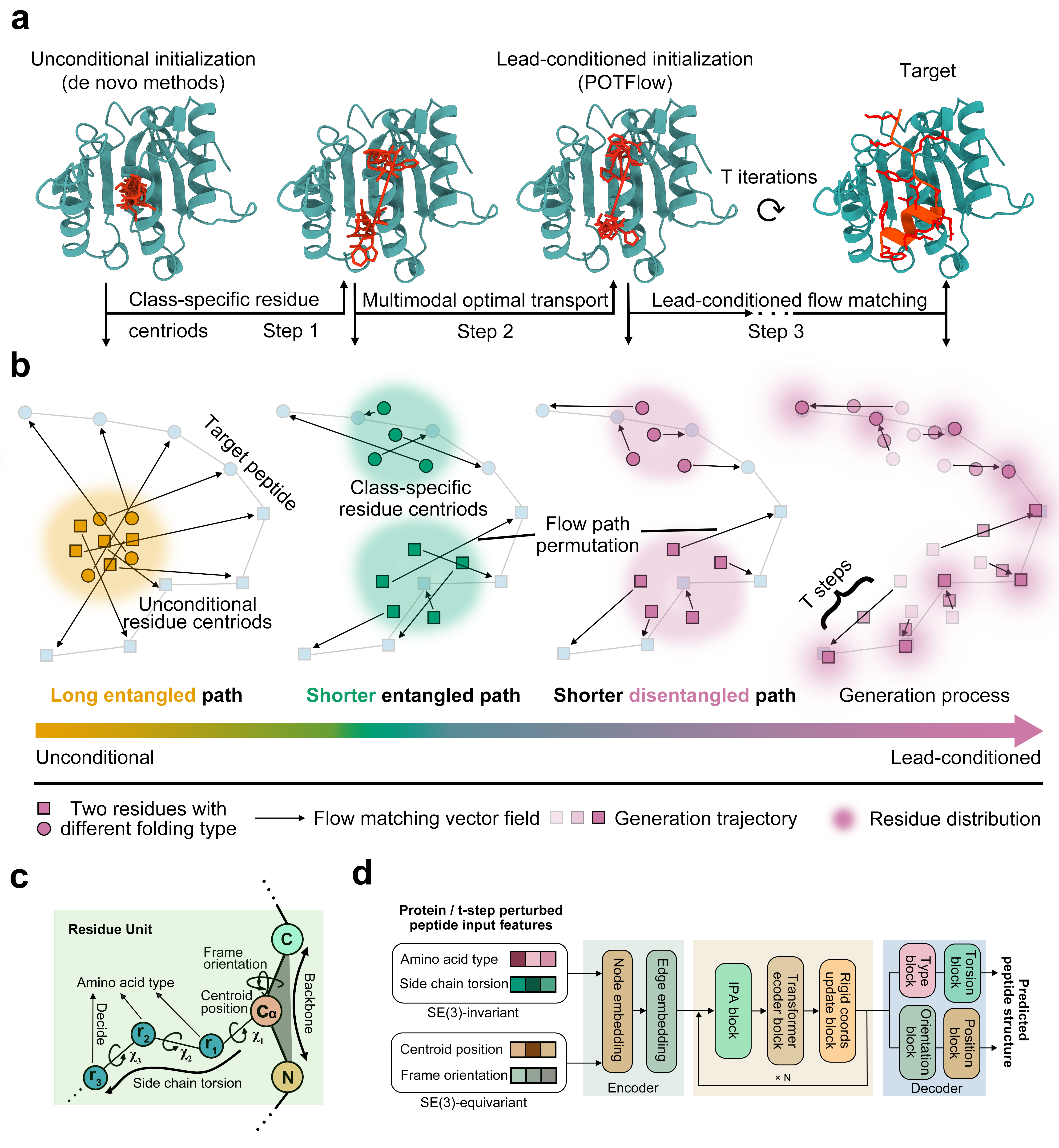}
    \caption{\textbf{a)} An overview of \textbf{POTFlow}. First, class-specific centroids are computed from lead peptide structures. Here, ``class'' represents the peptide folding type (e.g., helix, sheet, loop). Next, based on the optimal transport theory, multimodal couplings between peptides and initial noise variables are established. Finally, short disentangled paths are built and lead-conditioned flow matching model generates high-affinity complex structures. \textbf{b)} Corresponding residue-level illustration of how POTFlow constructs more efficient generation trajectories via lead-conditioned initialization. \textbf{c)} Structural definition of a residue unit. \textbf{d)} Computational workflow of our model at $t$ time step.}
    
    \label{fig:POTFlow}
\end{figure}

Here, we explain how POTFlow incorporates the lead peptide as an additional input. 
As illustrated in the left plot in Fig.~\ref{fig:POTFlow}b, de novo methods first sample noisy peptides from an unconditional distribution. Time-independent vector fields are then learned to transport these samples along flow matching paths toward the target peptides. POTFlow reformulates this process by conditioning the learned vector fields on the structure of a lead peptide. Instead of conditioning the model by simply appending a fixed lead-peptide embedding to all inputs, POTFlow constructs a hierarchy of lead-aware representations that modulate the flow fields at multiple geometric and sequence levels.

Firstly, we transform the unconditional residue centriods into class-specific residue centriods, where each ``class'' corresponds to a peptide folding type (e.g., helix, sheet, or loop) defined by the lead peptide structure (Step 1 in Fig.~\ref{fig:POTFlow}a). Here, the centroid refers to the $C\alpha$ position of each residue (see Fig.~\ref{fig:POTFlow}c for details). We show that class-specific initialization leads to shorter flow matching paths (left and center-left panels in Fig.~\ref{fig:POTFlow}b), thereby improving generation efficiency. A rigorous proof is provided in Supplementary Notes. After this step, an ``improved'' noisy peptide is sampled. 

Secondly, we apply multimodal optimal transport (OT) between the lead peptide and the noisy peptide (Step 2 in Fig.~\ref{fig:POTFlow}a). As shown in Fig.~\ref{fig:POTFlow}c, a peptide consists of four multimodal representations, including residue type in discrete sequence space, centroid coordinates in Euclidean space, backbone frame orientations on the SO(3) manifold, and side-chain torsion angles on the torus manifold. For clarity, we only visualize the changes  of residue centroid paths after OT operation (center-left and center-right panels in Fig.~\ref{fig:POTFlow}b). The OT operation can be interpreted as a permutation that reassigns noisy residues to their target counterparts, constructing more efficient pairings and minimizing the overall length of flow matching paths. Moreover, optimal transport naturally yields disentangled flow paths, facilitating more stable and efficient ODE-based generation~\cite{tong2023improving}. Together, class-specific initialization and optimal transport-based residue pairing constitute a lead-conditioned scheme that guides the generation process along shorter and more disentangled flow trajectories. Our approach improves both the efficiency and stability of the ODE-based generative process, as formally justified in Supplementary Notes~\ref{app:proof}. 

Lastly, the aforementioned process produces paired noisy and lead peptides, which serve as anchors for the subsequent flow-based generation process. We then build a straight flow matching path by linearly interpolating within paired samples (Step 3 in Fig.~\ref{fig:POTFlow}a; rightmost panel in Fig.~\ref{fig:POTFlow}b). We refer to these procedures as Coupled Conditional Flow Matching (C$^2$FM). Full implementation details are provided in Methods section and Supplementary Materials~\ref{app:net}.

We then train a neural network $\phi_\theta$ to approximate the velocity field along the flow path. As illustrated in Fig.~\ref{fig:POTFlow}d, at an arbitrary timestep $t$, four residue-related variables are provided as input to $\phi_\theta$, which is optimized by minimizing the squared error against the lead peptide. During inference, starting from a noisy peptide sampled from the lead-conditioned initialization, we numerically integrate the corresponding ODE using the trained $\phi_\theta$, thereby transporting the system smoothly from noise to the desired peptide structure.

\begin{table*}[t]
\rowcolors{2}{white}{gray!10}
    \centering
\renewcommand{\arraystretch}{1.7}
\resizebox{\linewidth}{!}{
    \begin{tabular}{ccccccc}
        \toprule
        \multirow{2}{*}{Experiments} & \multicolumn{2}{c}{Distribution} & \multicolumn{2}{c}{Energy} & \multicolumn{2}{c}{Geometry} \\
        \cmidrule(r){2-3}
        \cmidrule(r){4-5}
        \cmidrule(l){6-7}
        \multicolumn{1}{c}{} & Similarity \% $\uparrow$ & Compactness \% $\uparrow$ & Affinity \% $\uparrow$ & Stability \% $\uparrow$ & RMSD \AA \ $\downarrow$  & BSR \% $\uparrow$ \\
        \midrule
        RFDiffusion      &46.26 &74.61 &16.53 &\textbf{26.82} &4.17  &26.71 \\
        ProteinGenerator &47.61 &77.43 &13.47 &23.84 &4.35  &24.62 \\
        PepFlow          &49.74 &79.77 &21.37 &18.15 &2.07  &86.89 \\
        PepGLAD          &24.93 &67.90 &10.47 &20.39 &3.83  &19.34 \\
        PepHAR           &20.89 &70.42 &20.53 &16.62 &2.68  &86.74 \\
        \midrule
        Vanilla      &32.20   &72.71 &18.37 &21.89 &2.44  &79.0 \\
        POTFlow w/o class prior  &51.96 &94.61 &24.21 &23.12 &1.77 &86.03 \\
        POTFlow w/o OT           &49.75 &81.13 &25.62 &26.35 &1.85 &\textbf{87.10} \\
        POTFlow                  &\textbf{53.44} &\textbf{95.07} &\textbf{30.56} &23.80 &\textbf{1.66} &87.01 \\
        \bottomrule
    \end{tabular}
}
    \vspace{1em}
    \caption{Comparison of five de novo peptide design methods and three POTFlow variants on peptide benchmark. Bold values denote the best performance for each metric.}
    \vspace{-.5em}

    \label{tab:benchmark}
\end{table*}

\subsection*{POTFlow outperforms existing methods on peptide benchmark}

To evaluate our model, we curated a benchmark dataset from PepBDB~\cite{wen2019pepbdb} and Q-BioLip~\cite{wei2024qbiolip} following prior work~\cite{li2024pepflow, li2024hotspot}, yielding 8,207 training samples and 158 test complexes. We evaluate POTFlow against five de novo generative methods, including diffusion-based models, i.e., RFDiffusion~\cite{watson2023rfd}, ProteinGenerator~\cite{lisanza2023pg}, and PepGLAD~\cite{kong2024pepglad}, a flow matching model, PepFlow~\cite{li2024pepflow} and an autogressive model PepHAR~\cite{li2024hotspot}. De novo methods explore the complete peptide space, they often suffer from inefficiency and low success rates due to the vast combinatorial landscape and lack of biological priors. In contrast, POTFlow generates candidates conditioned on a lead peptide, narrowing the search to regions more likely to yield bioactive candidates. Here, classifier-guided techniques~\cite{dhariwal2021diffusion,chen2025generating,qian2024kgdiff} are not considered, as they require iterative gradient-based refinement with external scoring functions. Such techniques are computationally expensive, vulnerable to biased scoring functions, and unstable during inference. 

As shown in Table~\ref{tab:benchmark}, POTFlow outperforms existing de novo peptide design methods across multiple evaluation metrics, demonstrating the effectiveness of its lead-conditioned generation strategy. In terms of distribution metrics, POTFlow achieves the highest Similarity and Compactness, indicating that it generates peptides structurally and sequentially closer to the lead peptide while maintaining a highly concentrated distribution in the latent space. This advantage results from the proposed C$\mathrm{^2}$FM technique, which effectively restricts the sampling process to a focused region around the lead peptide.

For energy metrics, POTFlow attains the highest Affinity score, reflecting its superior capability in generating peptides with stronger predicted binding to the target protein. In terms of geometric quality, POTFlow obtains the lowest RMSD, confirming its precision in producing structurally accurate peptide conformations. These improvements are attributed to the integration of class-specific priors and optimal transport–based residue pairing, which collectively construct shorter, disentangled flow paths for generation. POTFlow also achieves a high binding site rate (BSR), which measures the overlap between generated and native peptide binding interfaces. This highlights the effectiveness of lead-conditioned generation in capturing geometrically relevant interactions overlapping with the binding mode of lead peptides. For detailed metric definitions, please refer to the Methods section.

\begin{figure}[t] 
    \centering
    \includegraphics[width=\linewidth]{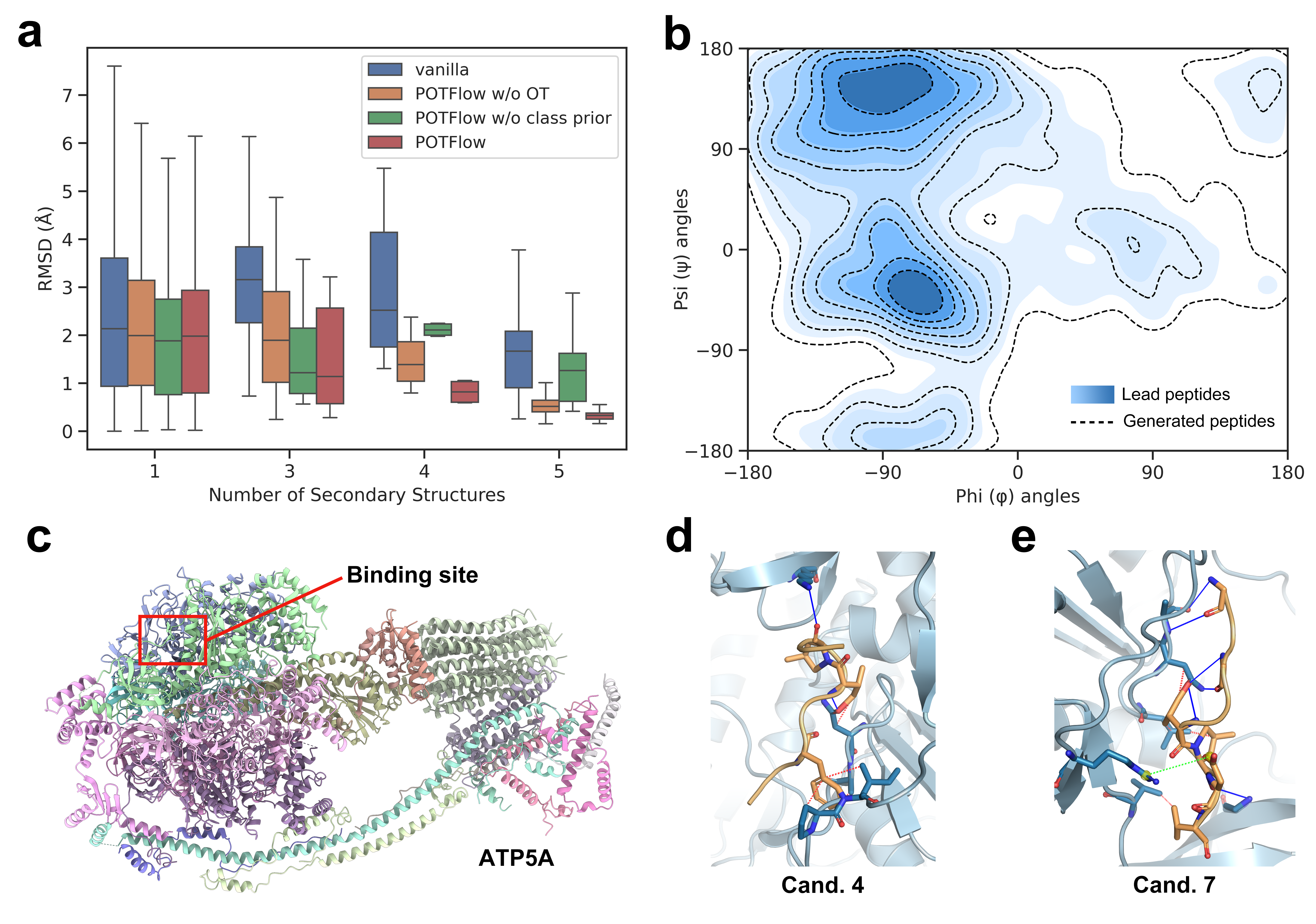}
    \caption{\textbf{a)} RMSD values on different number of secondary structures across four models. The RMSD values are computed between generated peptides and the lead peptides in the test set. \textbf{b)} Ramachandran plot of POTFlow generated and lead peptides. \textbf{c)} Visualization of the ATP5A subunit. The red box indicates the experimentally validated peptide-binding site used as the input for subsequent generative modeling. \textbf{d-e)} Detailed non-covalent interactions of two generated candidates with protein ATP5A, as identified by the Protein–Ligand Interaction Profiler (PLIP)~\cite{salentin2015plip}. Hydrogen bonds are shown as blue solid lines, hydrophobic contacts as red dashed lines, and salt bridges as green dashed lines.}
    \label{fig:rmsd_phipsi}
\end{figure}

\subsection*{Ablation studies highlight the effectiveness of POTFlow}

The last four rows of Table~\ref{tab:benchmark} present an ablation study of POTFlow, testing different variations and configurations of the model to assess their impact on various performance metrics. When class priors are removed (POTFlow w/o class prior), there is a noticeable drop in performance, especially for the affinity metric. Similarly, the removal of optimal transport policy (POTFlow w/o ot) also leads to reduced performance, particularly in terms of compactness and similarity metrics. When both class priors and optimal transport policy are removed (Vanilla), the performance drops to the worst across most metrics.

\begin{figure}[htbp]
    \centering
    \includegraphics[width=\linewidth]{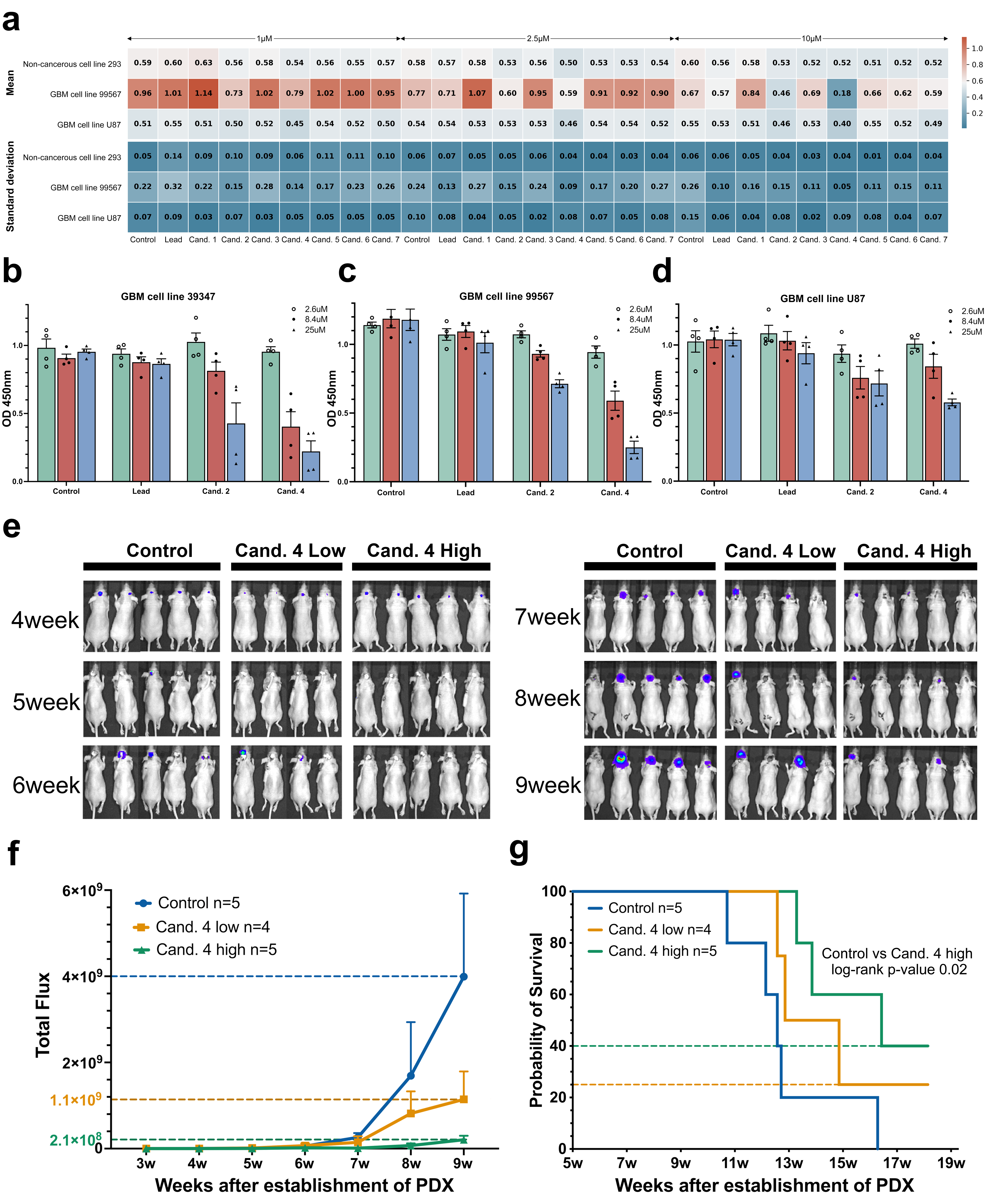}
    \caption{ \textbf{a)}: Heatmap of cell viability (mean $\pm$ SE) in GBM and non-cancerous cells. Cells from non-cancerous cell line 293 and two GBM patient-derived lines (99567 and U87) were treated with control, a reference lead peptide, and seven candidate peptides (Cand.~1–7) at three concentrations (1 $\mu$M, 2.5 $\mu$M, and 10 $\mu$M). Viability was measured by CCK-8 assay and normalized to the control for each cell line. The upper panel shows the mean normalized viability, and the lower panel displays the corresponding standard deviation across nine independent experiments. Color intensity reflects relative viability (blue = lower, red = higher). \textbf{b-d)}: Cell viability of GBM cell lines (39347, 99567 and U87 respectively) with control, a lead peptide, Cand.~2 and Cand.~4 at 2.6 $\mu$M (open circles), 8.4 $\mu$M (filled circles) and 25 $\mu$M (triangles). The y-axis represents absorbance at 450 nm . \textbf{e)} Bioluminescence imaging of tumor progression in mice under different treatment conditions. Images were taken weekly from week 4 to week 9. \textbf{f)} In vivo bioluminescence imaging of PDX mice following local delivery of Cand. 4 via subcutaneous minipumps. Mice received either control (blue circles, $n = 5$, 10mg/kg), low-dose Cand.~4 (yellow squares, $n = 4$, 10mg/kg), or high-dose Cand.~4 (green triangles, $n = 5$, 20mg/kg). \textbf{g)} Kaplan–Meier survival curves of PDX mice treated with the same regimens. Survival was monitored daily after tumor establishment. Here, $n$ indicates the number of mice per group.
    }
    \label{fig:wet_data}
\end{figure}

We further divide the peptides into four groups based on the number of secondary structures and compute the RMSD metric. As shown in Fig.~\ref{fig:rmsd_phipsi}a, POTFlow without optimal transport (w/o OT) consistently achieves lower RMSD values than the vanilla model across all groups. Notably, we observe a progressive reduction in RMSD as the number of secondary structure classes increases. This trend reflects the fact that more classes yield more informative priors, leading to shorter and more efficient flow matching paths, thereby improving model performance (Proposition~\ref{pro:CS} in Supplementary Notes~\ref{app:proof}).
Similarly, POTFlow without class priors also outperforms the vanilla model in all groups. Together, these results highlight that both optimal transport and class-specific priors are essential for improving overall performance (Theorem~\ref{thm:final} in Supplementary Notes~\ref{app:proof}).

Additionally, we plot the Ramachandran~\cite{RAMACHANDRAN196395} figure for both generated peptides and lead peptides. This plot illustrates the distribution of backbone dihedral angles $\psi$ and $\phi$, which are key determinants of protein conformation and folding behavior. As shown in Fig.~\ref{fig:rmsd_phipsi}b, the \(\phi\)–\(\psi\) angle distribution of the generated candidates closely matches that of the lead peptides, indicating realistic backbone conformations.

\subsection*{Optimization of the ATP5A-binding lead peptide via POTFlow}

Recent work by You et al.~\cite{You_qtd} delineated the DDIT4L--TOM40--ATP5A pathway as an endogenous brake on GBM oncogenesis and showed that a DDIT4L-derived peptide (DDIT4L$^{\mathrm{V125\text{--}P132}}$) targeting ATP5A can impair mitochondrial ATP synthase activity, induce GBM cell apoptosis, and suppress tumor growth in orthotopic PDX models. However, this ATP5A-binding peptide represents only a first-generation lead obtained by motif selection, and its sequence space and therapeutic window remain largely unexplored. In this study, we aim to optimize this existing ATP5A-binding lead peptide using our geometry-aware generative framework POTFlow. As described below, POTFlow transforms this suboptimal lead into structurally diverse candidates and identifies peptides that not only enhance cell viability inhibition and tumor selectivity in vitro but also significantly reduce tumor burden and prolong survival in a GBM PDX model.

\subsection*{Rapid design and screening of peptides for GBM in silico}

% \subsection*{Rapid design and screening of peptides for GBM in silico} 
% The previous study~\cite{You_qtd} has identified ATP5A as a GBM-relevant target and reported an ATP5A-binding lead peptide with limited efficacy. In this study, we do not aim to rediscover ATP5A as a target, but rather to systematically optimize an existing ATP5A-binding lead peptide using our geometry-aware generative framework. In the following sections, we will describe how POTFlow transforms this suboptimal lead into optimized candidates and identifies peptides that not only improve cell viability inhibition and tumor selectivity in vitro but also significantly suppress tumor growth and prolong survival in a PDX model.

Starting from the lead peptide sequence which binds to ATP5A within a pocket (Fig.~\ref{fig:rmsd_phipsi}c), we employed our framework to further optimize this candidate. Specifically, given the 3D structure of ATP5A and the lead peptide sequence, we initialize peptide–protein binding conformations using AlphaFold 3. Based on these conformations, POTFlow performs lead-conditioned flow matching in a learned latent space to efficiently explore the local structural neighborhood of the lead. This process yields 2,985 unique peptide candidates with binding modes informed by the original lead. To bridge generative modeling and experimental validation, we implemented a multi-dimensional expert system that leverages sequence similarity, predicted affinity, and interaction profiles to navigate the lead-conditioned design space constructed by POTFlow. First, we clustered the peptides into 15 groups using GibbsCluster~\cite{Andreatta2017gibbs}, and selected the group whose representative sequence exhibited the highest similarity to the lead peptide. Within this group, we then evaluated 110 peptides in terms of predicted binding affinity and peptide–protein interaction profiles. Based on this integrative assessment, seven candidates were finally selected for downstream experimental validation.

To explore the binding potential of POTFlow-generated peptides, we visualized two candidates for their interaction profiles. As shown in Fig.~\ref{fig:rmsd_phipsi}d-e, Cand.~4 engages the binding site through three hydrogen bonds and two hydrophobic contacts, while Cand.~7 exhibits a denser interaction network including five hydrogen bonds, three hydrophobic contacts, and a stabilizing salt bridge. These visualizations suggest that POTFlow-generated peptides can establish diverse and potentially stable non-covalent interactions with the target.
% The bioactivity of the selected peptides was further validated experimentally through cell viability assays and PDX models, as described in the following sections.

% \subsubsection{Related Works}

% \subsubsection{Ablation studies}
% Table~\ref{tab:ablation} presents an ablation study of POTFlow, testing different variations and configurations of the model to assess their impact on various performance metrics.

% When class priors are removed (POTFlow w/o class prior), there is a noticeable drop in performance, especially in \textbf{Affinity}. Similarly, the removal of optimal transport (POTFlow w/o ot) also leads to reduced performance, particularly in terms of \textbf{Compactness} and \textbf{Similarity}. When both class priors and optimal transport (except for optimal transport on $C_\alpha$) are removed, the performance drops to the lowest across most metrics.

% This analysis emphasizes the crucial role of both class priors and optimal transport in improving the performance of POTFlow, making it a more effective tool for peptide design.

\subsection*{Improved cell viability inhibition and tumor selectivity of generated peptides} 

% \begin{table}[ht]
%     \centering
% \resizebox{\linewidth}{!}{
% \renewcommand{\arraystretch}{1.3}
%     \begin{tabular}{c|ccc}
%         \toprule
%         % \multirow{2}{*}{Concentration} & \multicolumn{2}{c}{Cell line 293} & \multicolumn{2}{c}{Cell line 99567} & \multicolumn{2}{c}{Cell line U87} \\
%         % \cmidrule(r){2-3}
%         % \cmidrule(r){4-5}
%         % \cmidrule(l){6-7}
%         %  & drop \% $\downarrow$ & prop. \% $\uparrow$
%         %  & drop \% $\downarrow$ & prop. \% $\uparrow$
%         %  & drop \% $\downarrow$ & prop. \% $\uparrow$\\
%         Concentration & Cell line 293 (\%) $\uparrow$& Cell line 99567 (\%) $\uparrow$ & Cell line U87 (\%) $\uparrow$\\
%         \hline
%         1 $\mu$M      & 85.71 & 57.14 & 100 \\
%         2.5 $\mu$M        & 85.71  &28.57 & 100 \\
%         10 $\mu$M        & 85.71 & 28.57 & 85.71 \\
%         \midrule
%         \multicolumn{4}{r}{Average improvement: 73.02} \\
%         \bottomrule
%     \end{tabular}
% }
%     \caption{
% Proportion of GBMCurer-generated peptides that outperform the lead peptide in inhibiting cell viability, across three cell lines (293, 99567, U87) and three concentrations (1, 2.5, and 10~$\mu$M). The average improvement rate across all tested conditions is 73.02\%.
% }
%     \label{tab:viability_percent}
% \end{table}

Seven candidates were evaluated for their ability to inhibit cell viability in three different cell lines. These candidates were tested at three concentrations (1, 2.5, and 10~$\mu$M) on two GBM cell lines (99567 and U87) and one non-cancerous cell line (293). Cell viability was quantified using Cell Counting Kit-8 (CCK-8) assay. This colorimetric method measures mitochondrial dehydrogenase activity in viable cells, with lower absorbance values at 450~nm indicating reduced cell viability and increased peptide-induced cytotoxicity.

Fig.~\ref{fig:wet_data}a shows the detailed inhibitory effects of the candidates across all cell lines and concentrations after 24 hours of treatment. Our candidates showed stronger inhibition of cell viability compared to the lead peptide in most of the test conditions. Notably, Cand.~2 and Cand.~4 showed strongest inhibitory effects than the other candidates. We further validated the inhibitory effects of Cand.~2 and Cand.~4 on tumor cells under higher seeding density conditions. As shown in Fig.~\ref{fig:wet_data}b–d, both candidates effectively suppressed the viability of all three GBM cell lines after 24 hours of treatment. In addition to reduced cell viability, they also led to decreased intracellular ATP levels, indicating potential inhibition of ATP synthase activity (Supplementary Fig.~\ref{fig_app:atp}). These results highlight the robustness of POTFlow in lead-conditioned peptide design.
% As shown in Table~\ref{tab:viability_percent}, the candidates outperformed the lead peptide in 73.02\% of the tested conditions. These results demonstrate the capability of POTFlow to consistently generate peptides that surpass the lead sequence under diverse experimental settings. 

To evaluate peptide selectivity, we calculated the inhibition of viability rate (IVR):
\[\text{IVR} = \frac{\text{Viability}_{\text{Lead}} - \text{Viability}_{\text{Candate}}}{\text{Viability}_{\text{Lead}}} \times 100\%.
\]
Notably, Cand. 4 exhibited a significantly higher inhibitory rate in GBM cells (18$\sim$68\%) than in non-cancerous cell (less than 10\%, Supplementary Table~\ref{tab:cell_selectivity}).
% Notably, we found that Cand.~4 demonstrated strong cell line selectivity, with a 27.4\% improvement in GBM cells compared to only 9.8\% in non-cancerous cells. 
These results highlight that our framework not only enhances cell viability inhibition but also yields candidates with improved specificity toward GBM cells.

\subsubsection*{Improved tumor suppression by peptide candidate 4 in the PDX model}

% To establish a PDX model, tumor specimens were obtained from a GBM patient via sterile surgical resection. The specimens were enzymatically dissociated and then cultured in neural stem cell medium for four passages to enrich tumor-initiating cells. These cells were transduced with a lentiviral vector encoding firefly luciferase, enabling in vivo bioluminescent tracking. Luciferase-expressing cells were orthotopically implanted into the brains of immunocompromised nude mice. 
We established a GBM PDX model by orthotopically implanting luciferase-labeled, patient-derived tumor-initiating cells into immunocompromised mice. Cand.~4, which most effectively suppressed tumor cell viability in vitro, was administered intracranially by osmotic pumps (ALZET®, 1002) to direct deliver the peptide or control peptide (dissolved in PBS) to the injection site 4 weeks after intracranial tumor cell injection. Tumor progression was monitored weekly via in vivo bioluminescence using the IVIS Lumina II system. 

As shown in Fig.~\ref{fig:wet_data}e-f, high-dose Cand.~4 treatment (20mg/kg) led to a marked suppression of tumor growth in PDX mice, with bioluminescent signals remaining low throughout the 11-week observation period. By contrast, mice in control group exhibited exponential tumor expansion beginning at week 7, whereas the low-dose group (10mg/kg) displayed a moderate inhibitory effect. Survival analysis further revealed that all control mice succumbed by week 16, while 25\% and 40\% of mice in the low- and high-dose groups, respectively, survived beyond week 18 (Fig.~\ref{fig:wet_data}g). Notably, high-dose Cand.~4 significantly improved survival compared to control (\emph{p-value} = 0.02, Supplementary Table~\ref{tab:Log-rank}). These results demonstrate the in vivo efficacy of Cand.~4 and highlight its potential as a therapeutic peptide for GBM.

\section*{Discussion}
In this study, we present a dry-to-wet peptide design framework for GBM. Rather than generating peptides from scratch, our model POTFlow adopts a lead-conditioned scheme to guide generation. By leveraging the lead peptide, POTFlow efficiently explores its local sequence–structure space to produce candidates with enhanced therapeutic potential. We validated the framework in both cellular and PDX models, demonstrating its ability to generate peptides with tumor selectivity in vitro and therapeutic efficacy in vivo. Although POTFlow was applied to generate peptides only for GBM in this study, it can be readily adapted to other diseases by providing the corresponding target protein and the lead peptide. 

While POTFlow achieves state-of-the-art performance in our benchmarks, several opportunities for further improvement remain. First, the current framework does not explicitly consider peptide properties, such as solubility and metabolic stability, which are critical for practical application. Second, while our approach is validated through both computational and experimental evaluations, the current workflow lacks an iterative feedback mechanism: experimental results are not yet integrated into the generative process to guide subsequent optimization. Third, while we observed ATP depletion and robust anti-tumor effects in vitro and in vivo, we have not yet directly visualized changes in ATP5A subcellular localization (e.g., by IF) or fully dissected the in vivo mechanism beyond survival and tumor burden. These aspects represent promising directions for improving the performance and robustness of our framework.

In summary, we present the first generalizable framework for lead-conditioned peptide design with experimental validation. Our approach outperforms the lead peptide in inhibiting cell viability in vitro and demonstrates in vivo efficacy in a PDX model. Combined with its adaptability to diverse targets, our framework offers a promising platform for therapeutic peptide design.

\section*{Materials and Methods}

\subsection*{Notations}

Protein-peptide complexes are represented as $\mathcal{C} = \{\mathcal{P}, \mathcal{G} \}$ consisting of protein $\mathcal{P}$ and peptide $\mathcal{G}$, both of which can be decomposed as residue frames. The geometric structure of the i-th residue is parameterized with $C_\alpha$ coordinate $\mathbf{x}^{(i)} \in \mathbb{R}^3$ and a frame orientation matrix $\mathbf{o}^{(i)} \in SO^3$ following AlphaFold 3~\cite{abramson2024af3}. We define i-th residue side-chain angles as $\mathbf{\chi}^{(i)} = \{\chi^{(i)}_1, \chi^{(i)}_2, \chi^{(i)}_3, \chi^{(i)}_4, \chi^{(i)}_5\}$ and its residue type as $\mathbf{c}^{(i)} \in \mathbb{R}^{20}$. As a result, a protein/peptide consisting of $N$ residues can be represented as $\{\mathcal{R}^{(i)}\}_{i=1}^{N}$, where $\mathcal{R}^{(i)} = \{\mathbf{x}^{(i)}, \mathbf{o}^{(i)}, \mathbf{\chi}^{(i)}, \mathbf{c}^{(i)} \}$.

\subsection*{The Overview of POTFlow}
Given a target protein pocket $\mathcal{P}$ and a lead peptide $\mathcal{G}_{l}$, the task of POTFlow is to design new peptides $\mathcal{G}$ that not only bind effectively to the target protein but also closely resemble the lead peptide. Mathematically, the objective of POTFlow is to learn the dual-conditioned distribution $p_\theta(\mathcal{G} \mid \mathcal{P}, \mathcal{G}_{l})$ which can empirically be decomposed into four independent components:
\begin{equation}
    p(\mathcal{G} \,|\, \mathcal{P}, \mathcal{G}_{l}) \propto p(\{\mathbf{x}^{(i)}\}_{i=1}^{N}\mid \mathcal{P}, \mathcal{G}_{l}) \cdot p(\{\mathbf{o}^{(i)}\}_{i=1}^{N}\mid \mathcal{P}, \mathcal{G}_{l}) \cdot \nonumber p(\{\mathbf{\chi}^{(i)}\}_{i=1}^{N}\mid \mathcal{P}, \mathcal{G}_{l}) \cdot p(\{\mathbf{c}^{(i)}\}_{i=1}^{N}\mid \mathcal{P}, \mathcal{G}_{l}).
\end{equation}

Our approach differs from previous methods~\cite{li2024pepflow, lin2024ppflow, kong2024pepglad, wang2024target, li2024hotspot} by incorporating an additional condition based on the lead peptide, which enhances practicality in real-world peptide design. To effectively integrate this new variable, we introduce two techniques to address our objective: 1) constructing the prior distribution of residue centroids based on the secondary structure of the lead peptide; 2) applying a multimodal optimal transport policy to the initialized noisy residues for shorter and disentangled flow matching paths. We name the integration of Conditional Flow Matching with the coupling technique as C${^2}$FM. 

\subsection*{Prior distribution of residue centriods}
% The secondary structure of peptides directly affects their biological activity, stability, and binding affinity with target proteins~\cite{bechinger1993structure, romeo1988structure,zasloff2002antimicrobial}. Using the $\texttt{DSSP}$ module from Biopython~\cite{cock2009biopython}, we classify residues into three types of secondary structures. Please refer to Appendix~\ref{table:ss_table} for more details. We integrate this domain knowledge into our method by initializing the $C_\alpha$ positions of residues according to a prior distribution. This distribution is derived from classifying the residues in the lead peptide into secondary structure motifs:

% We give detailed proofs for Proposition~\ref{pro:RE} and Proposition~\ref{pro:CS} in Appendix~\ref{app:proof} and a toy example in Appendix~\ref{toy:ss}.

The secondary structure of a peptide plays a pivotal role in determining its biological function, structural stability, and binding interactions with target proteins~\cite{bechinger1993structure, romeo1988structure, zasloff2002antimicrobial}. To incorporate this key biological insight into our generative process, we leverage secondary structure annotations derived from the DSSP algorithm, implemented via the Biopython library~\cite{cock2009biopython}. Each residue in the lead peptide is assigned to a structural class, such as helix, sheet, or loop. The complete mapping is summarized in Supplementary Table~\ref{table:ss_table}.

To initialize the 3D coordinates of backbone $C_\alpha$ atoms in a biologically meaningful way, we construct a prior distribution conditioned on the secondary structure class of each residue. Specifically, for each structural class, we compute the average 3D position (centroid) of the residues assigned to that class in the input structure. These centroids serve as class-specific anchors. The new initial positions are then sampled from a normal distribution centered at the corresponding centroid, ensuring that the structural context is preserved during initialization.

This strategy contrasts with an unconditional initialization approach, which treats all residues identically (i.e., ``global initialization''). We found that using class-specific centroids leads to a more accurate reconstruction of the original geometry, as it reduces the average deviation between initialized and reference coordinates. Moreover, this sampling process is equivariant to rotation. That is, if the input structure is rotated in 3D space, the initialized positions will rotate in a consistent manner. This property is essential for preserving the geometric integrity of protein-peptide complexes during downstream modeling. Further details are provided in Supplementary Notes~\ref{app:proof} and Supplementary Fig.~\ref{fig:global_vs_class}.

\subsection*{Multimodal C$^2$FM technique}
% Recall that a peptide can be represented by \(N\) residues, with each residue characterized by four key features: residue $C_\alpha$ positions, residue orientations, side-chain torsions, and residue types. These features exist in different manifolds and occupy unique dimensional spaces with distinct properties which makes peptide data multimodal. In the following subsections, we describe our method across the four geometric components of peptide structure in detail. 
% \begin{equation} \label{eq:ot}
% \pi^* = \arg\min_{\pi \in \Pi(\mathbf{X}_0, \mathbf{X}_1)} \sum_{i,j} \pi_{ij} \cdot d\left(\mathbf{x}_0^{(i)}, \mathbf{x}_1^{(j)}\right),
% \end{equation}

Each peptide can be described as a sequence of \(N\) residues, where each residue has four main features: $C_\alpha$ position, backbone orientation, side-chain torsion angles, and amino acid type. These features lie in distinct geometric spaces, making peptide modeling a multimodal problem. Below, we introduce our approach for each component within the C$^2$FM framework.

\subsubsection*{C$^2$FM for $C_\alpha$ positions}

The 3D coordinates of the $C_\alpha$ atoms form the backbone scaffold of a peptide. We denote the position of the $i$-th residue as $\mathbf{x}^{(i)} \in \mathbb{R}^3$, and collectively as $\mathbf{X} = \{\mathbf{x}^{(i)}\}_{i=1}^{N}$. Initial positions are sampled from a class-conditional Gaussian prior informed by secondary structure annotations ( Supplementary Eq.~\ref{eq:class-specific}). A continuous flow is then trained to transform these noisy initializations $\mathbf{X}_0$ toward ground-truth structures $\mathbf{X}_1$ by learning the velocity field along a linear interpolation:
\begin{equation}
    \mathbf{X}_t = (1 - t)\,\mathbf{X}_0 + t\,\mathbf{X}_1,
\end{equation} where $t \in [0,1]$.
A neural network $v_t^{pos}$ is trained to minimize the flow-matching loss between predicted and true velocity fields $\mathbf{X}_1-\mathbf{X}_0$.

To better align $\mathbf{X}_0$ with $\mathbf{X}_1$, we introduce a structure-aware optimal transport~\cite{villani2021topics,rao2019engineering} (OT) step. OT provides a principled way to find a minimal-cost assignment between two sets of points under a given distance metric. In our case, the cost is defined as the squared Euclidean distance between $C_\alpha$ positions, and the transport is constrained to occur only within the same secondary structure class.

Formally, we solve a constrained OT problem to obtain a discrete permutation $\Pi^{pos}$ that minimizes the total transport cost:
\begin{equation}\label{eq:s2ot_pos}
    W_2(\mathbf{X}_0, \mathbf{X}_1) = \sum_{i,j} \Pi^{pos}_{i,j} \left\|\mathbf{x}_0^{(i)} - \mathbf{x}_1^{(j)}\right\|^2.
\end{equation} 
This yields a reassigned initialization $\mathbf{X}_0' \sim \Pi^{pos}(X_0, X_1)$ that provides shorter and more stable flow trajectories during training. The OT solution is efficiently computed using the Python Optimal Transport (\texttt{POT}) library~\cite{flamary2021pot}, benefiting from the limited length of peptide sequences (typically $N < 25$). Unlike previous work~\cite{tong2023improving}, which applies OT between samples to improve inter-sample flow trajectories, our approach performs OT within each individual peptide, i.e., intra-sample flow trajectories. 

During inference, new $C_\alpha$ positions are generated by numerically integrating the learned velocity field starting from $\mathbf{X}_0'$ using forward Euler steps:
\begin{equation}\label{eq:sample_pos}
    \mathbf{X}_{t+\Delta t}' = \mathbf{X}_t' + \Delta t \cdot v_t^{pos}(\mathbf{X}_t'; \theta).
\end{equation}

\subsubsection*{\textbf{C$^2$FM for frame orientations}}

We represent the orientation of i-th residue as a rotation matrix $\mathbf{o}^{(i)} \in SO(3)$, with the full set denoted as $\mathbf{O} = \{\mathbf{o}^{(i)}\}_{i=1}^{N}$. Unlike Euclidean space, $SO(3)$ is a curved manifold, making standard linear interpolation inapplicable. To generate smooth and valid orientation trajectories, we construct geodesic flows on $SO(3)$ using operations in its Lie algebra. Initial orientations $\mathbf{o}_0^{(i)}$ are sampled uniformly from $SO(3)$, and $\mathbf{o}_1^{(i)}$ are taken from ground-truth structures.

We define the flow path as the geodesic between the two rotations:
\begin{equation}
\mathbf{O}_t = \exp_{\mathbf{O}_0} \left(t \cdot \log_{\mathbf{O}_0}(\mathbf{O}_1)\right), t\in[0,1].
\end{equation}

To reduce path length and improve flow stability, we compute an optimal transport permutation $\Pi^{ori}$ that minimizes the total transport cost on Riemannian manifold:
\begin{equation}\label{eq:s2ot_ori}
W_{\mathfrak{so}(3)}(\mathbf{O}_0, \mathbf{O}_1) = \sum_{i,j} \Pi^{ori}_{i,j} \left\| \log_{\mathbf{o}_0^{(i)}}(\mathbf{o}_1^{(j)}) \right\|^2.
\end{equation}

This results in a permuted initialization $\mathbf{O}_0' \sim \Pi^{ori}(\mathbf{O}_0, \mathbf{O}_1)$ that more closely aligns with the target orientations. A neural network $v_t^{ori}$ is then trained to predict the geodesic velocity vector $\log_{\mathbf{O}_0'}(\mathbf{O}_1)$ along the path.

During generation, we apply geodesic Euler updates:
\begin{equation}\label{eq:sample_ori}
\mathbf{O}_{t+\Delta t}' = \exp_{\mathbf{O}_t'} \left( \Delta t \cdot v_t^{ori}(\mathbf{O}_t'; \theta) \right).
\end{equation}

\subsubsection*{\textbf{C$^2$FM for residue types}}
The amino acid types in a peptide can be denoted as $C = \{c^{(i)}\}_{i=1}^{N}$ where $c^{(i)} \in \{ c \in \mathbb{Z}^+ \mid 1 \leq c \leq 20 \} $. Here, we utilize a soft one-hot technique to map discrete $c^{(i)}$ into continuous one by $\texttt{onehot}_{\texttt{soft}}(c^{(i)}) = \mathbf{s}^{(i)} \in \mathbb{R}^{20}$, and the j-th value in $\mathbf{s}^{(i)}$ is defined as follows:
    \begin{equation}
        s^{(i)}[j] = \begin{cases} 
                    K, & j = c^{(i)} - 1 \\
                    -K, & j \ne c^{(i)} - 1 \\
                    \end{cases},
    \end{equation}
where $K$ is a constant. Here, $\mathbf{s}^{(i)}$ is treated as the logits of probabilities, and $\texttt{softmax}(\mathbf{s}^{(i)})$ represents the normalized distribution of residue types, where the $(c^{(i)}-1)$-th term closes to 1 and others close to 0. This means that $\texttt{softmax}(\mathbf{s}^{(i)})$ is a data point in 20-category probability simplex $\Delta^{19}$. We initialize $\mathbf{s}^{(i)} \sim \mathcal{N}(0, K^2I)$ so that the prior distribution on simplex becomes the logistic-normal distribution~\cite{hinde2011logistic, atchison1980logistic}.

The flow path is defined as a linear interpolation in logit space:
\begin{equation}
\mathbf{s}_t^{(i)} = (1 - t)\,\mathbf{s}_0^{(i)} + t\,\mathbf{s}_1^{(i)}, t\in[0,1].
\end{equation}

To better align initial and target residue distributions, we apply an optimal transport permutation $\Pi^{type}$ that minimizes the total cross-entropy cost between logit vectors:
\begin{equation}\label{eq:s2ot_type}
W_{CE}(\mathbf{S}_0,\mathbf{S}_1) = \sum_{i,j} \Pi^{type}_{i,j} \, \texttt{CE}(\mathbf{s}_0^{(i)}, \mathbf{s}_1^{(j)}),
\end{equation}
where $\texttt{CE}(\cdot,\cdot)$ denotes the cross-entropy between initial and true logits. The permuted initialization $\mathbf{S}_0' \sim \Pi^{type}(\mathbf{S}_0, \mathbf{S}_1)$ is used to train the neural network $v_t^{type}$ to approximate the logit velocity vector $\mathbf{S}_1 - \mathbf{S}_0'$.

During inference, we integrate the learned velocity in logit space and decode residue types as follows:
\begin{align}\label{eq:sample_type}
\mathbf{C}_{t+\Delta t}' &\sim \texttt{softmax}(\mathbf{S}_t' + \Delta t \cdot v_t^{type}(\mathbf{S}_t'; \theta)), \\
\mathbf{S}_{t+\Delta t}' &= \texttt{onehot}_{\texttt{soft}}(\mathbf{C}_{t+\Delta t}').
\end{align}

\subsubsection*{\textbf{C$^2$FM for side-chain torsions}}

Each residue $i$ has a set of side-chain torsion angles $\mathcal{X}^{(i)} = \{\chi^{(i,j)}\}_{j=1}^{N_i}$, where $N_i \in \{1,2,3,4,5\}$ depends on the residue type. Due to the periodic nature of torsional rotations, the side-chain angles $\chi^{(i,j)}$ are defined on the torus with periodicity $2\pi$ (or $\pi$ when rotational symmetry applies). Mathematically, the torsion vector $\mathcal{X}^{(i)}$ lies in the 5-dimensional torus manifold $\mathbb{T}^5$ which can also be interpreted as the quotient space $\mathbb{R}^5/(2\pi\mathbb{Z}^5)$. Exponential and logarithm maps on this Riemannian manifold are similar to those in Euclidean space, with the main difference being the equivalence relation that identifies elements differing by integer multiples of $2\pi$.

We initialize $\mathcal{X}_0^{(i)}$ from the uniform distribution over $[0, 2\pi)^{N_i}$ and define the flow path via wrapped linear interpolation:
\begin{equation}
\mathcal{X}_t^{(i)} = (t\mathcal{X}_1^{(i)} + (1-t)\mathcal{X}_0^{(i)}) \bmod 2\pi.
\end{equation}

We construct an optimal transport permutation $\Pi^{ang}$ using a cosine distance metric. Specifically, we embed each angle into its 2D representation $\vec{v}(\chi) = (\cos\chi, \sin\chi)$ and define the transport cost as:
\begin{equation}\label{eq:s2ot_ang}
W_{\text{torus}}(\mathcal{X}_0,\mathcal{X}_1) = \sum_{i,j} \Pi^{ang}_{i,j} \left\| \vec{v}(\chi_0^{(i)}) - \vec{v}(\chi_1^{(j)}) \right\|^2.
\end{equation}
A neural network $v_t^{ang}$ is trained using the permuted initialization $\mathcal{X}_0' \sim \Pi^{ang}(\mathcal{X}_0, \mathcal{X}_1)$ to learn wrapped velocity vectors $\texttt{wrap}(\mathcal{X}_1-\mathcal{X}_0')$, where $\texttt{wrap}(\cdot)$ maps the angle difference into $(-\pi, \pi]$ to ensure minimal angular displacement.

During inference, we generate torsion trajectories via wrapped Euler integration:
\begin{equation}\label{eq:sample_ang}
\mathcal{X}_{t+\Delta t}' = (\mathcal{X}_t' + \Delta t \cdot v_t^{ang}(\mathcal{X}_t'; \theta)) \bmod 2\pi.
\end{equation}

% This formulation respects the topology of the torus and enables smooth angular transitions across chemically valid conformational states.

% \subsection*{Summary of \textbf{POTFlow}}

% By integrating all multimodal flow matching objectives, the overall loss function is defined as follows:
% \begin{equation} \label{eq:all_loss}
%     \mathcal{L}_\theta^{all} = \lambda_1\mathcal{L}_\theta^{pos} + \lambda_2\mathcal{L}_\theta^{ori} + \lambda_3\mathcal{L}_\theta^{type} + \lambda_4\mathcal{L}_\theta^{ang},
% \end{equation}
% where $\lambda_1, \lambda_2, \lambda_3, \text{and} \ \lambda_4$ represent the respective weights of each component in the loss function. We provide complete pseudocode for training and sampling procedures of POTFlow in Supplementary information~\ref{pseudocode}.

\subsection*{Summary of \textbf{POTFlow}}

POTFlow unifies flow-based generative modeling across four distinct geometric spaces of peptide structure: Euclidean space for $C_\alpha$ positions, the rotation group $SO(3)$ for backbone orientations, the probability simplex for residue types, and the torus manifold for side-chain torsions. For each modality, we construct geometry-aware flow trajectories between noise and structure samples, initialized using structure-informed priors and refined via optimal transport to reduce path lengths.

The overall training objective minimizes the sum of squared velocity matching errors across all modalities, with tunable weights to balance the contribution of each component:
\begin{equation}\label{eq:all_loss}
    \mathcal{L}_{\text{total}} = \lambda_1 \cdot \text{Loss}_{\text{position}} + \lambda_2 \cdot \text{Loss}_{\text{orientation}} + \lambda_3 \cdot \text{Loss}_{\text{residue type}} + \lambda_4 \cdot \text{Loss}_{\text{torsion angle}}.
\end{equation}
At inference time, POTFlow generates peptide structures by integrating learned velocity fields from OT-aligned initializations via Euler steps, adapted to the geometry of each space. This enables coherent, physically consistent peptide structure generation across all structural modalities. Training and generation procedures are summarized in Supplementary Algorithm~\ref{pseudocode}.

\subsection*{Metrics}  
For the peptide optimization task, our goal is to design models that can efficiently explore the hidden space around lead peptides. The following metrics are used for evaluation:

\begin{enumerate}
    \item \textbf{Similarity} refers to the proportion of peptides that closely resemble the reference peptides, determined by two criteria: (a) a TM-score $\ge 0.5$~\cite{zhang2005tm,xu2010tm5}, and (b) a sequence identity $\ge 0.5$.
    \item \textbf{Compactness} measures the variability in both the structural and sequence features, calculated as 1 minus the product of pairwise (1-TM-score) and (1-sequence identity) across all peptides generated for a given target.
    \item \textbf{Affinity} evaluates the fraction of peptides that exhibit stronger binding affinities (or lower binding energies) compared to the native peptide.
    \item \textbf{Stability} assesses the percentage of the generated complexes that are more thermodynamically stable (having lower total energy) than the native complexes, using the Rosetta energy function~\cite{chaudhury2010pyrosetta}.
    \item \textbf{RMSD (Root-Mean-Square deviation)} computes the discrepancy between the generated peptide structures and the native structures by analyzing the $C_\alpha$ distances after alignment.
    % \item \textbf{SSR (Secondary structure ratio)} calculates the extent of overlap in secondary structures between the generated and native peptides, as identified by \texttt{DSSP}~\cite{kabsch1983dssp}.
    \item \textbf{BSR (Binding site rate)} quantifies the similarity in peptide-target interactions by evaluating the overlap of the binding sites.
\end{enumerate}

\subsection*{Baselines} 
We evaluate POTFlow against five powerful peptide design models. \textbf{RFDiffusion}~\cite{watson2023rfd} leverages pre-trained weights from RoseTTAFold~\cite{krishna2024rfaa} to generate protein backbone structures through a denoising diffusion process. The peptide sequences are subsequently reconstructed using \textbf{ProteinMPNN}~\cite{dauparas2022proteinmpnn}. \textbf{ProteinGenerator} enhances RFDiffusion by incorporating joint sequence-structure generation~\cite{lisanza2023pg}. \textbf{PepFlow}~\cite{li2024pepflow} generates full-atom peptides and samples them using a flow matching framework on a Riemannian manifold. \textbf{PepGLAD}~\cite{kong2024pepglad} utilizes equivariant latent diffusion networks to generate full-atom peptide structures. \textbf{PepHAR}~\cite{li2024hotspot} generates peptide residues autoregressively, based on a learned prior distribution for hotspot residues.

\section*{Data and code availability}

The data and code that support the findings of this study are available from the corresponding author upon reasonable request. All datasets and the complete codebase, including training and evaluation scripts, will be publicly released upon acceptance of the manuscript.

\bibliography{sn-bibliography}% common bib file

% \vspace{2em}
% \textbf{Hao Qian} is a PhD candidate in Computer Science at Shanghai Jiao Tong University. His research focuses on bioinformatics and machine learning.

\newpage

\appendix

\setcounter{figure}{0}
\setcounter{table}{0}
\section*{\centering Supplementary Materials for POTFlow}
\section{Notes} \label{app:proof}

\begin{definition} \label{def:ss}

We define a mapping function $f_{\texttt{DSSP}}$:
\begin{equation}
    f_\texttt{DSSP}(\{\mathbf{x}_1^{(i,j)}\}_{j=1,i=1}^{N_{ss},N_j}) \to \{g^{(j)}\}_{j=1}^{N_{ss}},
\end{equation}where $ f_{\texttt{DSSP}}(\mathbf{x}_1^{(i,j)}) = g^{(j)} $ specifies that the point $ \mathbf{x}_1^{(i,j)} $ is mapped to the class $ g^{(j)} $, and $ N_{ss} $ and $ N_j$ represents the number of secondary structures and the number of points in class $ g^{(j)} $ respectively.
\end{definition}
For each class $ g^{(j)} $, we determine its centroid $\mathbb{C}_{g^{(j)}}$ by averaging the positions of all points assigned to $ g^{(j)} $:
\begin{equation}
    \mathbb{C}_{g^{(j)}} = \frac{1}{N_j} \sum_{i=1}^{N_j} \mathbf{x}^{(i, j)}.
\end{equation}
Subsequently, we sample $ \mathbf{x}_0^{(i,j)} $ from a Gaussian distribution centered at the centroid $ \mathbb{C}_{g^{(j)}} $ with identity covariance:
\begin{equation} \label{eq:class-specific}
\mathbf{x}_0^{(i,j)} \mid g^{(j)} \sim \mathcal{N}(\mathbb{C}_{g^{(j)}}, I_3). 
\end{equation}

Based on Definition~\ref{def:ss} and the proposed initialization, we can directly derive the following propositions:
\begin{proposition} [Rotation Equivariance] \label{pro:RE}
    The sampling scheme in Definition~\ref{def:ss} is \textbf{rotation equivariant}, meaning that for any rotation matrix \(R \in SO(3)\), if we rotate the original data points by \(R\), the newly sampled points also rotate by \(R\) accordingly.
\end{proposition}

\begin{proposition} [Improved Initialization via Class-Specific Centroids] \label{pro:CS}
    Class-Specific Initialization yields a lower or equal average squared distance compared to Global Initialization. 
    Let \(\mathbb{C}_{\text{global}}\) denote the global centroid of all points: 
    \begin{equation}
        \mathbb{C}_{\text{global}} = \frac{1}{N} \sum_{i=1}^{N} \mathbf{x}_1^{(i)}.
    \end{equation}
    Define two initialization schemes for new points $ \mathbf{x}_0^{(i,j)} $ and $ \mathbf{x}_0^{(i)} $ respectively:
    \begin{enumerate}
        \item \textbf{Class-Specific:} $\mathbf{x}_0^{(i,j)} \mid g^{(j)} \sim \mathcal{N}\left(\mathbb{C}_{g^{(j)}}, I_3\right).$
        \item \textbf{Global:} $\mathbf{x}_0^{(i)} \sim \mathcal{N}\left(\mathbb{C}_{\text{global}}, I_3\right).$
    \end{enumerate}
    Then, the Class-Specific Initialization minimizes the average squared Euclidean distance between $\mathbf{x}_0$ and the original points $\mathbf{x}_1$ compared to the Global Initialization. Formally,
    \begin{equation}
        \frac{1}{N}  \sum_{j=1}^{N_{ss}} \sum_{i=1}^{N_j} \mathbb{E}\left[\|\mathbf{x}_0^{(i,j)} - \mathbf{x}_1^{(i,j)}\|^2\right] \Bigg|_{\text{Class-Specific}} \nonumber \\ \leq  \frac{1}{N} \sum_{i=1}^{N} \mathbb{E}\left[\|\mathbf{x}_0^{(i)} - \mathbf{x}_1^{(i)}\|^2\right] \Bigg|_{\text{Global}}.
    \end{equation}
        
\end{proposition}

\begin{theorem} \label{thm:final}
(\textbf{Shorter Paths via Class-Specific + OT Couplings})
Let \(\mathcal{X}\subset\mathbb{R}^d\) be the ambient space of a particular modality (e.g., \(C_\alpha\)-positions in Euclidean space). Suppose we have:
\[
   q_1(\mathbf{x}_1) : \text{target distribution},
   q_0^{\text{(vanilla)}}(\mathbf{x}_0): \text{vanilla prior}.
\]
We define:
\[
   q_0^{\text{(ours)}}(\mathbf{x}_0')
\]
to be the \emph{class-specific plus OT–based prior}, constructed by (1) centering initialization per class (Proposition~\ref{pro:CS}) and (2) coupling those initialization points to the data points via an \emph{optimal transport} policy \(\Pi^*\). Then
\[
   W^2\!\bigl(q_0^{\text{(ours)}},\,q_1\bigr)
   \;\;\leq\;\;
   W^2\!\bigl(q_0^{\text{(vanilla)}},\,q_1\bigr),
\]
where \(W^2(\cdot,\cdot)\) denotes the Wasserstein-2 (or geodesic) transport cost,\footnote{%
In Euclidean space, 
\(
   W^2(p,q)
   = 
   \inf_{\Pi\in\Gamma(p,q)}
   \mathbb{E}_{(\mathbf{x},\mathbf{y})\sim\Pi}
   \|\mathbf{x}-\mathbf{y}\|^2.
\)
In manifolds such as \(SO(3)\) or \(\mathbb{T}^k\), we replace \(\|\cdot\|^2\) by the corresponding geodesic-squared distance.%
}
and the inequality is strict if the vanilla prior \(q_0^{\text{(vanilla)}}\) differs sufficiently from \(q_0^{\text{(ours)}}.\)
\end{theorem}

\subsection{Proof of Proposition~\ref{pro:RE}}
\begin{proof}
    \textbf{1. Invariance of Class Labels Under Rotation.}  
    By assumption, the class labels \(\{g^{(j)}\}\) (e.g., secondary structure types) are determined by intrinsic properties such as local geometric configurations. Therefore, applying a uniform rotation \(R \in SO(3)\) to all data points does not alter their class memberships. Formally, if \(x_1^{(i,j)} \in g^{(j)}\), then \(R x_1^{(i,j)} \in g^{(j)}\).

    \textbf{2. Transformation of Centroids Under Rotation.}  
    After applying rotation \(R\) to each point in class \(g^{(j)}\), the new centroid \(\widetilde{\mathbb{C}}_{g^{(j)}}\) becomes
    \[
        \widetilde{\mathbb{C}}_{g^{(j)}} = \frac{1}{N_j} \sum_{i=1}^{N_j} R x_1^{(i,j)} = R \left( \frac{1}{N_j} \sum_{i=1}^{N_j} x_1^{(i,j)} \right) = R \mathbb{C}_{g^{(j)}}.
    \]
    Thus, the centroid undergoes the same rotation \(R\).

    \textbf{3. Rotation Invariance of Isotropic Gaussian Distributions.}  
    Each new point \(x_0^{(i,j)}\) is sampled from
    \[
        x_0^{(i,j)} = \mathbb{C}_{g^{(j)}} + \boldsymbol{\epsilon}^{(i,j)},
    \]
    where \(\boldsymbol{\epsilon}^{(i,j)} \sim \mathcal{N}(\mathbf{0}, I_3)\). Since the covariance matrix \(I_3\) is isotropic, it remains unchanged under rotation. Therefore, for any \(R \in SO(3)\),
    \[
        R \boldsymbol{\epsilon}^{(i,j)} \sim \mathcal{N}(\mathbf{0}, I_3).
    \]
    This implies that the distribution of the noise term \(\boldsymbol{\epsilon}^{(i,j)}\) is rotation invariant.

    \textbf{4. Overall Rotation Equivariance of the Sampling Scheme.}  
    Suppose the entire dataset is rotated by \(R \in SO(3)\). The rotated centroid is \(\widetilde{\mathbb{C}}_{g^{(j)}} = R \mathbb{C}_{g^{(j)}}\), and the rotated noise term is \(\widetilde{\boldsymbol{\epsilon}}^{(i,j)} = R \boldsymbol{\epsilon}^{(i,j)}\). The newly sampled point after rotation is
    \[
        \widetilde{x}_0^{(i,j)} = \widetilde{\mathbb{C}}_{g^{(j)}} + \widetilde{\boldsymbol{\epsilon}}^{(i,j)} = R \mathbb{C}_{g^{(j)}} + R \boldsymbol{\epsilon}^{(i,j)} = R \left( \mathbb{C}_{g^{(j)}} + \boldsymbol{\epsilon}^{(i,j)} \right) = R x_0^{(i,j)}.
    \]
    Therefore, the sampling scheme satisfies
    \[
        \widetilde{x}_0^{(i,j)} = R x_0^{(i,j)},
    \]
    which aligns with the definition of rotation equivariance.

\end{proof}

\subsection{Proof of Proposition~\ref{pro:CS}} \label{Improved_Initialization}
% Class-Specific Initialization yields a lower or equal average squared distance compared to Global Initialization. Let \(\mathbb{C}_{\text{global}}\) denote the global centroid of all points:
%     \[
%         \mathbb{C}_{\text{global}} = \frac{1}{N} \sum_{i=1}^{N} x_1^{(i)}.
%     \]
%     Define two initialization schemes for new points $ x_0^{(i,j)} $ and $ x_0^{(i)} $:
%     \begin{enumerate}
%         \item \textbf{Class-Specific Initialization:}
%         \[
%             x_0^{(i,j)} \mid g^{(j)} \sim \mathcal{N}\left(\mathbb{C}_{g^{(j)}}, I_3\right).
%         \]
%         \item \textbf{Global Initialization:}
%         \[
%             x_0^{(i)} \sim \mathcal{N}\left(\mathbb{C}_{\text{global}}, I_3\right).
%         \]
%     \end{enumerate}
%     Then, the Class-Specific Initialization minimizes the average squared Euclidean distance between \(x_0\) and the original points \(x_1\) compared to the Global Initialization. Formally,
%     \[
%         \frac{1}{N} \sum_{j=1}^{N_{ss}} \sum_{i=1}^{N_j} \mathbb{E}\left[\|x_0^{(i,j)} - x_1^{(i,j)}\|^2\right] \Bigg|_{\text{Class-Specific}} \leq \frac{1}{N} \sum_{i=1}^{N} \mathbb{E}\left[\|x_0^{(i)} - x_1^{(i)}\|^2\right] \Bigg|_{\text{Global}}.
%     \]

\begin{proof}
    To compare the two initialization schemes, we compute the expected squared Euclidean distance between the initialized points \(x_0\) and the original points \(x_1\) under both schemes.

    \textbf{1. Class-Specific Initialization:}
    \[
        x_0^{(i,j)} = \mathbb{C}_{g^{(j)}} + \boldsymbol{\epsilon}^{(i,j)},
    \]
    where \(\boldsymbol{\epsilon}^{(i,j)} \sim \mathcal{N}(\mathbf{0}, I_3)\).

    The squared distance is:
    \[
        \|x_0^{(i,j)} - x_1^{(i,j)}\|^2 = \|\mathbb{C}_{g^{(j)}} - x_1^{(i,j)} + \boldsymbol{\epsilon}^{(i,j)}\|^2.
    \]
    Expanding the square and taking expectation:
    \[
        \mathbb{E}\left[\|x_0^{(i,j)} - x_1^{(i,j)}\|^2\right] = \|\mathbb{C}_{g^{(j)}} - x_1^{(i,j)}\|^2 + \mathbb{E}\left[\|\boldsymbol{\epsilon}^{(i,j)}\|^2\right] + 2 \mathbb{E}\left[ (\mathbb{C}_{g^{(j)}} - x_1^{(i,j)})^\top \boldsymbol{\epsilon}^{(i,j)} \right].
    \]
    Since \(\boldsymbol{\epsilon}^{(i,j)}\) has zero mean and is independent of \(x_1^{(i,j)}\),
    \[
        \mathbb{E}\left[ (\mathbb{C}_{g^{(j)}} - x_1^{(i,j)})^\top \boldsymbol{\epsilon}^{(i,j)} \right] = 0.
    \]
    Also, \(\mathbb{E}\left[\|\boldsymbol{\epsilon}^{(i,j)}\|^2\right] = \text{trace}(I_3) = 3\). Therefore,
    \[
        \mathbb{E}\left[\|x_0^{(i,j)} - x_1^{(i,j)}\|^2\right] = \|\mathbb{C}_{g^{(j)}} - x_1^{(i,j)}\|^2 + 3.
    \]
    Taking the average over all points:
    \begin{align} 
    \begin{aligned}
        \frac{1}{N} \sum_{j=1}^{N_{ss}} \sum_{i=1}^{N_j} \mathbb{E}\left[\|x_0^{(i,j)} - x_1^{(i,j)}\|^2\right]_{\text{Class-Specific}} = \frac{1}{N} \sum_{j=1}^{N_{ss}} \sum_{i=1}^{N_j} \left( \| \mathbb{C}_{g^{(j)}} - x_1^{(i,j)} \|^2 + 3 \right) \nonumber \\
        = \frac{1}{N} \sum_{j=1}^{N_{ss}} \sum_{i=1}^{N_j} \| \mathbb{C}_{g^{(j)}} - x_1^{(i,j)} \|^2 + 3.
    \end{aligned}
    \end{align} 

    \textbf{2. Global Initialization:}
    \[
        x_0^{(i)} = \mathbb{C}_{\text{global}} + \boldsymbol{\epsilon}^{(i)},
    \]
    where \(\boldsymbol{\epsilon}^{(i)} \sim \mathcal{N}(\mathbf{0}, I_3)\).

    Similarly,
    \[
        \mathbb{E}\left[\|x_0^{(i)} - x_1^{(i)}\|^2\right] = \|\mathbb{C}_{\text{global}} - x_1^{(i)}\|^2 + 3.
    \]
    Taking the average over all points:
    \begin{align} 
    \begin{aligned}
        \frac{1}{N} \sum_{i=1}^{N} \mathbb{E}\left[\|x_0^{(i)} - x_1^{(i)}\|^2\right]_{\text{Global}} = \frac{1}{N} \sum_{i=1}^{N} \left( \| \mathbb{C}_{\text{global}} - x_1^{(i)} \|^2 + 3 \right) \\ \nonumber
        = \frac{1}{N} \sum_{i=1}^{N} \| \mathbb{C}_{\text{global}} - x_1^{(i)} \|^2 + 3.
    \end{aligned}
    \end{align} 
    
    \textbf{3. Comparing the Two Schemes:}

    To show that Class-Specific Initialization yields a smaller average squared distance, it suffices to show:
    \[
        \frac{1}{N} \sum_{j=1}^{N_{ss}} \sum_{i=1}^{N_j} \| \mathbb{C}_{g^{(j)}} - x_1^{(i,j)} \|^2 \leq \frac{1}{N} \sum_{i=1}^{N} \| \mathbb{C}_{\text{global}} - x_1^{(i)} \|^2.
    \]
    This follows from the \textbf{Law of Total Variance}~\cite{weiss2006course} or the \textbf{Within-Class Sum of Squares Minimization} in clustering.

    \begin{lemma} \label{lemma:tss}
        Let \(\{x_1^{(i,j)}\}\) be partitioned into classes \(\{g^{(j)}\}\). Then,
        \[
            \frac{1}{N} \sum_{j=1}^{N_{ss}} \sum_{i=1}^{N_j} \| \mathbb{C}_{g^{(j)}} - x_1^{(i,j)} \|^2 \leq \frac{1}{N} \sum_{i=1}^{N} \| \mathbb{C}_{\text{global}} - x_1^{(i)} \|^2.
        \]
    \end{lemma}

    \begin{proof}[Proof of Lemma]
        Consider the decomposition of variance:
        \[
            \frac{1}{N} \sum_{i=1}^{N} \| x_1^{(i,j)} - \mathbb{C}_{\text{global}} \|^2 = \frac{1}{N} \sum_{j=1}^{N_{ss}} \sum_{i=1}^{N_j} \| x_1^{(i,j)} - \mathbb{C}_{g^{(j)}} \|^2 + \frac{1}{N} \sum_{j=1}^{N_{ss}} N_j \| \mathbb{C}_{g^{(j)}} - \mathbb{C}_{\text{global}} \|^2.
        \]
        This is known as the \textbf{Total Sum of Squares (TSS)}~\cite{mardia1979tss} decomposition into \textbf{Within-Cluster Sum of Squares (WCSS)} and \textbf{Between-Cluster Sum of Squares (BCSS)}:
        \[
            \text{TSS} = \text{WCSS} + \text{BCSS}.
        \]
        Since \(\text{BCSS} \geq 0\), it follows that:
        \[
            \text{WCSS} \leq \text{TSS}.
        \]
        Therefore,
        \[
            \frac{1}{N} \sum_{j=1}^{N_{ss}} \sum_{i=1}^{N_j} \| x_1^{(i,j)} - \mathbb{C}_{g^{(j)}} \|^2 \leq \frac{1}{N} \sum_{i=1}^{N} \| x_1^{(i)} - \mathbb{C}_{\text{global}} \|^2.
        \]
    \end{proof}

    Adding the constant term \(3\) to both sides preserves the inequality:
    \[
        \frac{1}{N} \sum_{j=1}^{N_{ss}} \sum_{i=1}^{N_j} \left( \| \mathbb{C}_{g^{(j)}} - x_1^{(i,j)} \|^2 + 3 \right) \leq \frac{1}{N} \sum_{i=1}^{N} \left( \| \mathbb{C}_{\text{global}} - x_1^{(i)} \|^2 + 3 \right).
    \]
    Hence,
    \[
        \frac{1}{N} \sum_{j=1}^{N_{ss}} \sum_{i=1}^{N_j} \mathbb{E}\left[\|x_0^{(i,j)} - x_1^{(i,j)}\|^2\right] \Bigg|_{\text{Class-Specific}} \leq \frac{1}{N} \sum_{i=1}^{N} \mathbb{E}\left[\|x_0^{(i)} - x_1^{(i)}\|^2\right] \Bigg|_{\text{Global}}.
    \]
\end{proof}

\subsection{Proof of Theorem~\ref{thm:final}}
\begin{proof}
We prove this in three stages:

\paragraph{Step 1: Class-Specific Centers Reduce Within-Class Distances.}
By Proposition~\ref{pro:CS} in the main text, partitioning data points into secondary-structure classes and centering the initialization \(\mathbf{x}_0^{(i,j)}\) around \emph{class-specific} centroids \(\mathbb{C}_{g^{(j)}}\) lowers the expected squared distance to the target \(\mathbf{x}_1^{(i,j)}\) compared with using a \emph{global} centroid. Concretely:
\[
  \mathbb{E}\Bigl[\bigl\|\mathbf{x}_1^{(i,j)}-\mathbf{x}_0^{(\mathrm{class})}\bigr\|^2\Bigr]
  \;\le\; 
  \mathbb{E}\Bigl[\bigl\|\mathbf{x}_1^{(i,j)}-\mathbf{x}_0^{(\mathrm{global})}\bigr\|^2\Bigr],
\]
and summing over all classes and residues yields a strictly smaller or equal total distance.

\paragraph{Step 2: OT Coupling Minimizes Transport Cost.} \label{thm:step2}
Recall that the Wasserstein-2 distance can be written as an infimum over all couplings \(\Pi\) with marginals \(q_0\) and \(q_1\). The \emph{optimal} transport coupling \(\Pi^*\) exactly minimizes
\[
   \mathbb{E}_{(\mathbf{x}_0,\mathbf{x}_1)\sim \Pi}\bigl\|\mathbf{x}_1 - \mathbf{x}_0\bigr\|^2.
\]
Hence, using \(\Pi^*\) to re-initialize \(\mathbf{x}_0'\) \emph{must} yield a smaller expected distance to data than any naive (e.g., i.i.d.) pairing.

\paragraph{Step 3: Combining Both for a Tighter Prior.}
\begin{itemize}
    \item \textbf{Class-Specific Initialization} ensures each \(\mathbf{x}_0^{(i,j)}\) starts closer to \(\mathbf{x}_1^{(i,j)}\) within its semantic or structural class than a global scheme would.
    \item \textbf{OT Coupling} then permutes/rearranges these points to optimally match them to the data distribution, further lowering the transport cost.
    \item \textbf{Resulting Improvement:} the distribution \(q_0^{\text{(ours)}}\) of these re-initialized \(\mathbf{x}_0'\) satisfies
    \[
       W^2\!\bigl(q_0^{\text{(ours)}},\,q_1\bigr)
       \;\leq\;
       W^2\!\bigl(q_0^{\text{(class)}},\,q_1\bigr)
       \;\leq\;
       W^2\!\bigl(q_0^{\text{(vanilla)}},\,q_1\bigr).
    \]
\end{itemize}
Thus, our initialization strategy is strictly closer (in the Wasserstein sense) to \(q_1\) than a vanilla prior, yielding \emph{shorter probability paths} for Flow Matching.
\end{proof}

\begin{corollary}[Efficiency of Flow Matching] \label{corollary:efficiency}
The reduced Wasserstein distance enables:
\begin{itemize}
    \item Fewer ODE steps \(T\) for comparable sample quality
    \item Lower variance in vector field estimation
    \item Higher fidelity and coverage, since less “unnecessary traveling” occurs in high-dimensional spaces.
\end{itemize}
\end{corollary}

\section{Toy Examples}
\subsection{Global Initialization VS Class-Specific Initialization}

\begin{figure}[ht]
    \centering
    \includegraphics[width=\linewidth]{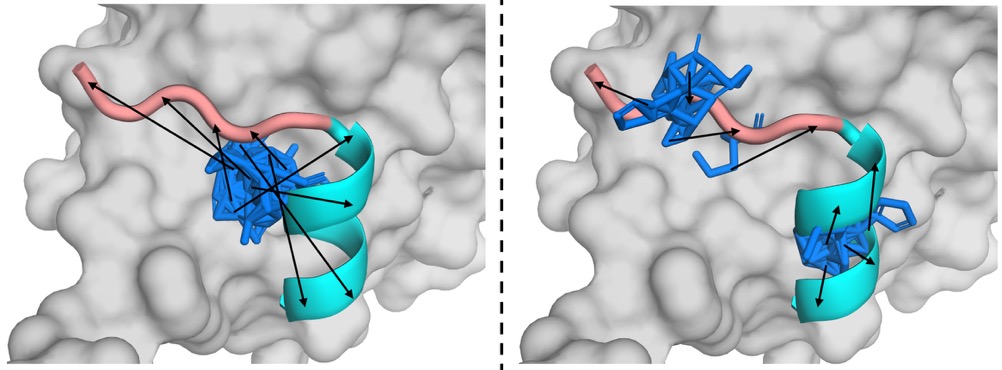}
    \caption{Comparison between global initialization (left) and class-specific initialization (right). The black arrows indicate the flow velocity vectors. As proved in Supplementary Notes~\ref{Improved_Initialization}, the flow trajectories in the right diagram are shorter than those in the left diagram.}
    \label{fig:global_vs_class}
\end{figure}
% While the amino acid data exists in a high-dimensional space, Figure~\ref{fig:global_vs_class} provides a simplified illustration of the concept. For theoretical proofs, please refer to Appendix .

\subsection{Optimal Transport between 2D Empirical Distributions}

As illustrated in Fig.~\ref{fig:ot_vs_w/o}, the trajectories in the right image show fewer cross-interactions than those in the left image. In flow matching models, ODE trajectories should not intersect during training. If they do, the model may learn a vector field that represents a compromise between the intersecting trajectories, blending their directions. This ambiguity in the vector field can prevent the model from accurately capturing the true flow dynamics, resulting in poor generalization and inaccurate behavior at the intersection points.

\begin{figure}[ht]
    \centering
    \includegraphics[width=\linewidth]{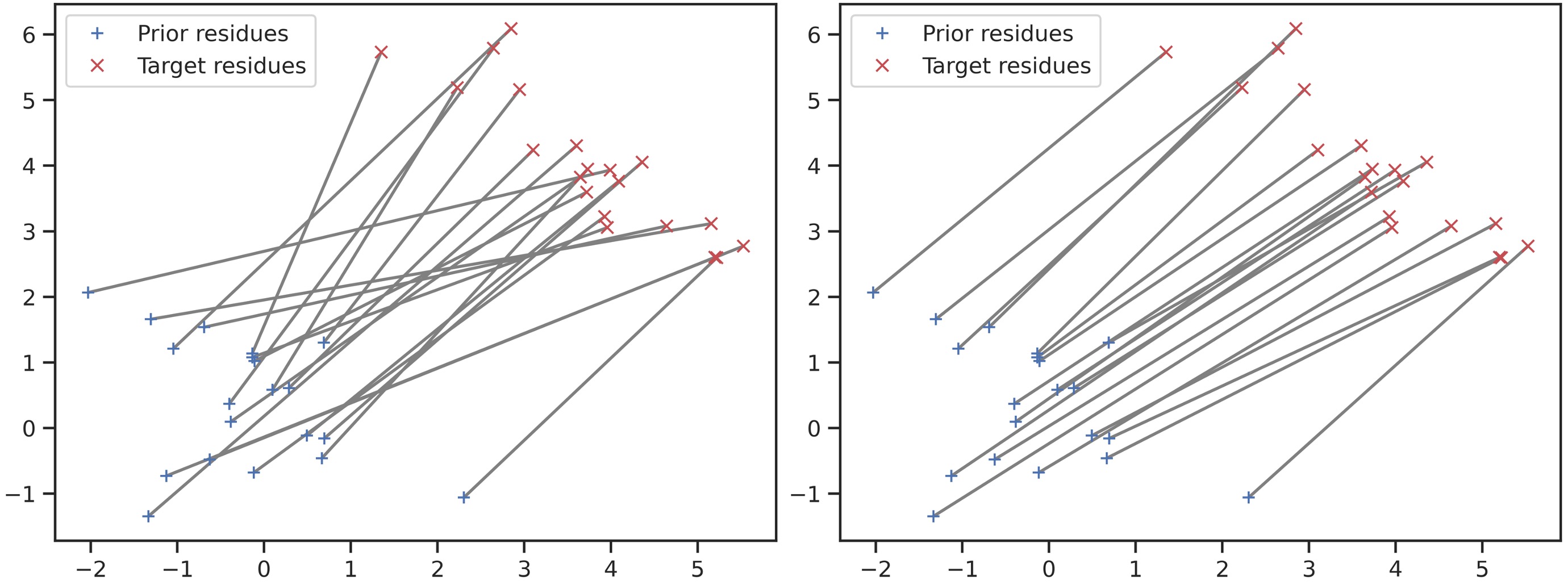}
    \caption{A toy example using 2D points to illustrate the different flow trajectories before (left) and after (right) applying the optimal transport plan.}
    \label{fig:ot_vs_w/o}
\end{figure}

% \subsection{Pre-process of Targeting the Glioblastoma Protein}

% \section{Dataset Details}
% \subsection{Overview of Peptide-protein Complex Dataset}
% We built our training and testing datasets, following previous works~\cite{li2024pepflow, li2024hotspot}. This benchmark, consisting of moderate-length sequences, is curated from PepBDB~\cite{wen2019pepbdb} and Q-BioLip~\cite{wei2024qbiolip}, with duplicates and low-quality data points excluded. The binding pocket is defined as the set of residues in the target protein, where the heavy atoms are within a 10 \AA \ radius of any heavy atom in the peptide. The dataset includes 158 complexes, grouped into 10 clusters identified using mmseqs2~\cite{steinegger2017mmseqs2}, along with 8207 examples used for training and validation.

\section{Classification of Peptide Secondary Structures}
Fig.~\ref{fig:ss_number} presents the empirical cumulative distribution function (ECDF) of the total number of secondary‐structure elements (helices, $\beta$‐strands, turns, etc.) per peptide candidate. The steep rise between 3 and 6 elements shows that most peptides contain 3–6 secondary‐structure features, with the cumulative proportion approaching unity as the count increases. Besides, as illustrated in Table~\ref{table:ss_table}, we list the eight DSSP secondary‐structure codes and their definitions, which we have consolidated into three primary categories: Type 0 (helical structures: H, G, I), Type 1 (strand structures: B, E), and Type 2 (turn/coil structures: T, S, –). This grouping preserves the main helix and strand classes while combining all flexible or unstructured segments into a single category to facilitate subsequent statistical analyses.

\begin{figure}[ht]
    \centering
    \includegraphics[width=0.5\linewidth]{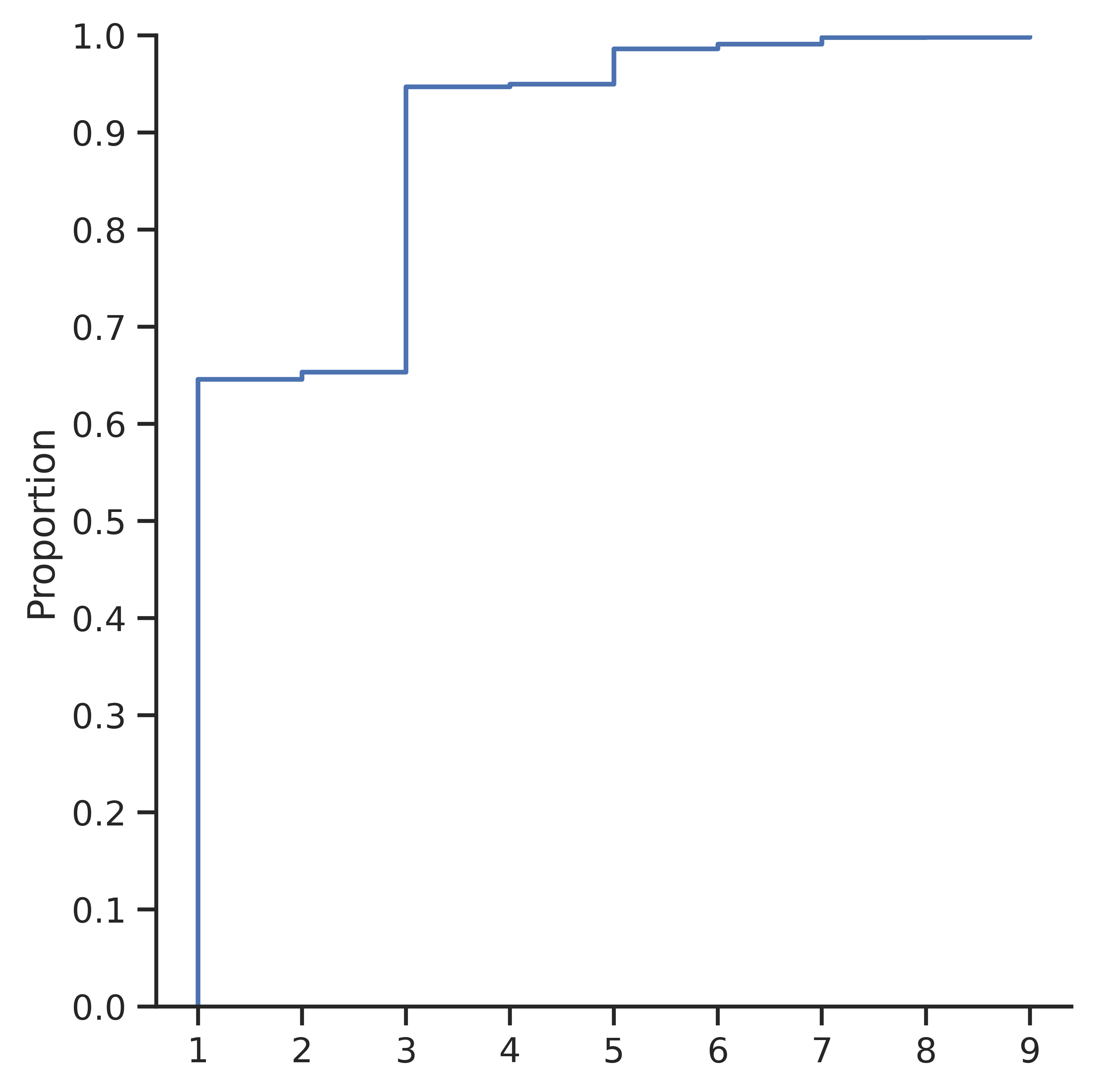}
    \caption{Empirical Cumulative Distribution Function (ECDF) analysis of the distribution of secondary structure counts in individual peptides.}
    \label{fig:ss_number}
\end{figure}

\begin{table*}[]
\renewcommand{\arraystretch}{1.5}
\centering
\begin{tabular}{|c|c|c|}
\hline
\textbf{Code} & \textbf{Structure} & \textbf{Type Index} \\ \hline
H & Alpha helix & 0 \\ \hline
B & Isolated beta-bridge residue & 1 \\ \hline
E & Strand & 1 \\ \hline
G & 3-10 helix & 0 \\ \hline
I & Pi helix & 0 \\ \hline
T & Turn & 2 \\ \hline
S & Bend & 2 \\ \hline
- & None & 2 \\ \hline
\end{tabular}
\caption{Eight types of secondary structures. For simplicity, we consolidate these into three primary categories, designated as 0, 1, and 2.}
\label{table:ss_table}
\end{table*}

\section{Experiment Setup}
\subsection{Model Architecture} \label{app:net}
Our flow model is primarily developed based on the guidelines set forth in PepFlow~\cite{li2024pepflow} and the Invariant Point Attention (IPA) architecture as detailed by Jumper et al.~\cite{jumper2021highly}, which utilizes both the provided features and the structural frames of the backbone. This model employs an invariant attention mechanism to analyze the interactions between the receptor protein and the peptide backbone currently under consideration.

Here, we give a detailed description of the model architecture.

\subsubsection*{Input Definitions}

Given a protein structure consisting of $N$ residues and a batch size of $B$, the input to the encoder includes:

\begin{itemize}
    % \item Initial node features $\mathbf{h}_i^{\text{init}} \in \mathbb{R}^{d_s}$ for residue $i$.
    \item Discrete amino acid sequence types $\text{seq}_i \in \{1,\dots,22\}$.
    \item Scalar timestep $t$ used in denoising diffusion models.
    \item Side chain torsion angles $\boldsymbol{\theta}_i \in \mathbb{R}^{5}$ ($\chi_1$, $\chi_2$, $\chi_3$, $\chi_4$, $\chi_5$).
    \item Rigid-body transformations $(\mathbf{R}_i^{(0)}, \mathbf{t}_i^{(0)})$ for each residue $i$.
\end{itemize}

\subsection*{NodeEmbedder}
For each residue \(i\):
\[
\mathbf{h}_i^{\mathrm{init}}
= \mathrm{MLP}_{\mathrm{node}}\bigl(\bigl[e_i^{\mathrm{aa}}
  \;\Vert\;
  e_i^{\mathrm{crd}}
  \;\Vert\;
  e_i^{\mathrm{dihed}}\bigr]\bigr)
\;\in\;\mathbb{R}^{d_s}\,,
\]
where:
\begin{itemize}
  \item \(e_i^{\mathrm{aa}}\) is the learned embedding of the amino-acid type.
  \item \(e_i^{\mathrm{crd}}\) are the local (C\(_\alpha\)-frame) atom coordinates flattened by one-hot selection.
  \item \(e_i^{\mathrm{dihed}}\) is the angular encoding of \((\phi,\psi,\omega)\).
\end{itemize}
The output is masked by whether the C\(_\alpha\) exists.

\subsection*{EdgeEmbedder}
For each residue pair \((i,j)\):
\[
\mathbf{Z}_{ij}^{(0)}
= \mathrm{MLP}_{\mathrm{edge}}\bigl(\bigl[e_{ij}^{\mathrm{aapair}}
  \;\Vert\;
  e_{ij}^{\mathrm{relpos}}
  \;\Vert\;
  e_{ij}^{\mathrm{dist}}
  \;\Vert\;
  e_{ij}^{\mathrm{dihed}}\bigr]\bigr)
\;\in\;\mathbb{R}^{d_z}\,,
\]
where:
\begin{itemize}
  \item \(e_{ij}^{\mathrm{aapair}}\) embeds the pair of residue types.
  \item \(e_{ij}^{\mathrm{relpos}}\) embeds their relative sequence offset (clamped).
  \item \(e_{ij}^{\mathrm{dist}}\) projects Gaussian-kernelized inter-atomic distances.
  \item \(e_{ij}^{\mathrm{dihed}}\) encodes the two inter-residue dihedrals.
\end{itemize}
The output is masked to same-chain residue pairs.

\subsection*{Initial Feature Mixing}
For each residue \(i\), fuse node, sequence, timestep and angle embeddings:
\[
\mathbf{H}_i^{(0)}
= \mathrm{MLP}\Bigl(\bigl[
    \mathbf{h}_i^{\mathrm{init}}
    \;\Vert\;
    \mathbf{s}_i
    \;\Vert\;
    \mathbf{t}_i
    \;\Vert\;
    \mathbf{a}_i
  \bigr]\Bigr)
\;\in\;\mathbb{R}^{d_s}\,,
\]
where:
\begin{itemize}
  \item \(\mathbf{s}_i\) is the learned embedding of the \emph{current} amino-acid at step \(t\).
  \item \(\mathbf{t}_i\) is the time-step embedding.
  \item \(\mathbf{a}_i\) is the angular encoding of the sampled torsions.
  \item \(\Vert\) denotes concatenation.
\end{itemize}

% \subsubsection*{Input Embedding and Feature Mixing}

% The initial embedding combines structural, sequential, temporal, and angular information as follows:

% \begin{align}
%     \mathbf{s}_i &= \text{Embed}(\text{seq}_i) \in \mathbb{R}^{d_s}, \\
%     \mathbf{t}_i &= \text{TimeEmbed}(t) \in \mathbb{R}^{d_s}, \\
%     \mathbf{a}_i &= \text{AngleEncode}(\boldsymbol{\theta}_i) \in \mathbb{R}^{d_a}, \\
%     \mathbf{H}_i^{(0)} &= \text{FFN}_{\text{mix}}\left([\mathbf{h}_i^{\text{init}} \Vert \mathbf{s}_i \Vert \mathbf{t}_i \Vert \mathbf{a}_i] \right) \in \mathbb{R}^{d_s}.
% \end{align}

% Here, $\text{FFN}_{\text{mix}}$ denotes a two-layer feedforward neural network with ReLU activations and $\Vert$ represents feature concatenation.

\subsubsection*{Iterative Block Updates}

Let $B$ denote the total number of blocks. Each block $b = 0, \dots, B-1$ performs the following computations:

\begin{equation*}
    \Delta \mathbf{H}_i^{\text{IPA}} = \text{IPA}^{(b)}(\mathbf{H}_i^{(b)}, \mathbf{Z}_{ij}^{(b)}, (\mathbf{R}_i^{(0)}, \mathbf{t}_i^{(0)})), 
\end{equation*}

\subsubsection*{Transformer Sequence Attention}

The output of IPA is combined with a skip connection and passed through a Transformer encoder:

\begin{align}
    \mathbf{X}_i^{(b)} &= \text{LayerNorm}(\mathbf{H}_i^{(b)} + \Delta \mathbf{H}_i^{\text{IPA}}), \\
    \mathbf{X}_i^{(b)} &= \text{TransformerEncoder}^{(b)}(\mathbf{X}_i^{(b)}), \\
    \mathbf{H}_i^{(b+1)} &= \text{MLP}(\mathbf{H}_i^{(b)} + \text{Linear}(\mathbf{X}_i^{(b)})).
\end{align}

\subsubsection*{Feedforward Transition and Rigid Update}

Node features are further processed and used to compute rigid-body updates:

\begin{align}
    \Delta \mathbf{r}_i^{(b)} &= \text{RigidUpdate}(\mathbf{H}_i^{(b+1)}), \\
    (\mathbf{R}_i^{(b+1)}, \mathbf{t}_i^{(b+1)}) &= (\mathbf{R}_i^{(b)}, \mathbf{t}_i^{(b)}) \circ \Delta \mathbf{r}_i^{(b)},
\end{align}

where $\circ$ denotes the composition of rigid transformations.

\subsubsection*{Edge Feature Update}

For all blocks except the last, edge embeddings are updated:

\begin{align}
    \mathbf{Z}_{ij}^{(b+1)} = \text{EdgeTransition}(\mathbf{H}_i^{(b+1)}, \mathbf{Z}_{ij}^{(b)}).
\end{align}

\subsubsection*{Final Output Predictions}

After $B$ blocks of iteration, the final outputs include:

\begin{align}
    \hat{\mathbf{R}}_i &= \mathbf{R}_i^{(B)}, \\
    \hat{\mathbf{t}}_i &= \mathbf{t}_i^{(B)}, \\
    \hat{\boldsymbol{\theta}}_i &= \text{AngleNet}(\mathbf{H}_i^{(B)}) \bmod 2\pi, \\
    \hat{\mathbf{p}}_i &= \text{SeqNet}(\mathbf{H}_i^{(B)}),
\end{align}

where $\hat{\mathbf{R}}_i$ and $\hat{\mathbf{t}}_i$ define the predicted rigid-body transformation of residue $i$, $\hat{\boldsymbol{\theta}}_i$ denotes the predicted side chain torsion angles, and $\hat{\mathbf{p}}_i$ is the predicted amino acid probability distribution.

\subsection{Training}
The distinctions between the three models are specified as follows: the class priors are omitted in the \textbf{POTFlow w/o class prior}; the \textbf{POTFlow w/o OT} is characterized by the exclusion of the optimal transport policy; and \textbf{POTFlow} represents the comprehensive model that incorporates all features, including the optimal transport policy and the class prior. The parameters are set as follows: $\lambda_1 = 0.5, \lambda_2 = 0.5, \lambda_3 = 1.0, \lambda_4 = 1.0$.

Each of these models undergoes training on four NVIDIA A100 GPUs, employing a Distributed Data Parallel (DDP) training framework over 2 million iterations. We maintain a learning rate of $5\times10^{-4}$ and configure the batch size to 16 for each GPU in the distributed setup.

\subsection{Sampling}
The sampling process is conducted on a single NVIDIA A100 GPU, utilizing 200 timesteps during the Euler step update. For each protein in the test set, we concurrently sample 64 peptides.

\section{Pseudocode} \label{pseudocode}

We outlined the overall training and sampling procedures here.
\begin{algorithm}[ht]
\caption{POTFlow Training Procedure}
\label{alg:training}
\begin{algorithmic}[1]
\State \textbf{Input:} Peptide dataset $D = \{(\mathcal{P}^i, \mathcal{G}^i_1 )\}_{i=1}^{n}$, where $\mathcal{P}^i$ represents the protein target and $\mathcal{G}^i$ represents the peptide.
\While{not converged}
    \State Sample a protein-peptide pair $(\mathcal{P}^i, \mathcal{G}^i_1)$ from the dataset and a noise peptide $\mathcal{G}^i_0$ from the multimodal prior distributions.
    \State Compute the multimodal transport policy $\Pi$ for the pair $(\mathcal{G}^i_0, \mathcal{G}^i_1)$ using Equation \eqref{eq:s2ot_pos}, \eqref{eq:s2ot_ori}, \eqref{eq:s2ot_type}, and \eqref{eq:s2ot_ang}, then sample a new noise peptide $\mathcal{G}^{i'}_0$ from the policy $\Pi$.
    \State Sample $t \sim U(0,1)$, and compute a noisy peptide $\mathcal{G}^{i'}_t$ by interpolating between $\mathcal{G}^{i'}_0$ and the target peptide $\mathcal{G}^{i}_1$.
    \State Compute the flow vector fields for the interpolated peptide $\mathcal{G}_t^{i'}$ using the neural network $v_t(\mathcal{G}_t^{i'}, \mathcal{P}; \theta)$.
    \State Compute the loss function $\mathcal{L}_\theta^{all}$ using Eq. \eqref{eq:all_loss} and update the model parameters $\theta$ accordingly.
\EndWhile
\State \textbf{Return:} Final model parameters $\theta$.
\end{algorithmic}
\end{algorithm}

\begin{algorithm} 
\caption{POTFlow Sampling Procedure}
\label{alg:sampling}
\begin{algorithmic}[1]
\State \textbf{Input:} The lead peptide $\mathcal{G}_{l}$ and its binding target protein $\mathcal{P}$.
\State Initialize $\mathcal{G}_0$ from the multimodal prior distributions.
\State Compute optimal transport map $\Pi$ between the peptide pair $(\mathcal{G}_0, \mathcal{G}_{l})$.
\State Sample $\mathcal{G}_0' \sim \Pi$.
\For{$t = 0$ \textbf{to} $1$ \textbf{step} $\frac{1}{T}$}
    \State Caculate the flow vector fields predicted by the neural network $v_t(\mathcal{G}_t', \mathcal{P}; \theta)$.
    \State Update the noisy peptide $\mathcal{G}_t'$ by the Euler methods in Equation \eqref{eq:sample_pos}, \eqref{eq:sample_ori}, \eqref{eq:sample_type} and \eqref{eq:sample_ang}.
\EndFor
\State \textbf{Return:} Final generated peptides $\mathcal{G}_1'$.
\end{algorithmic}
\end{algorithm}

\newpage
\section{Visualization} \label{sup:vis}
As shown in Fig.~\ref{app:vis}, we selected six peptides of different lengths from the test set for visualization. We found that shorter peptides closely match the lead peptide in sequence, which correlates with lower RMSD values. In contrast, longer peptides demonstrate increased sequence diversity and higher RMSD values, a reasonable outcome given their greater degrees of freedom. 

\begin{figure}[htb]
    \centering
    \includegraphics[width=\linewidth]{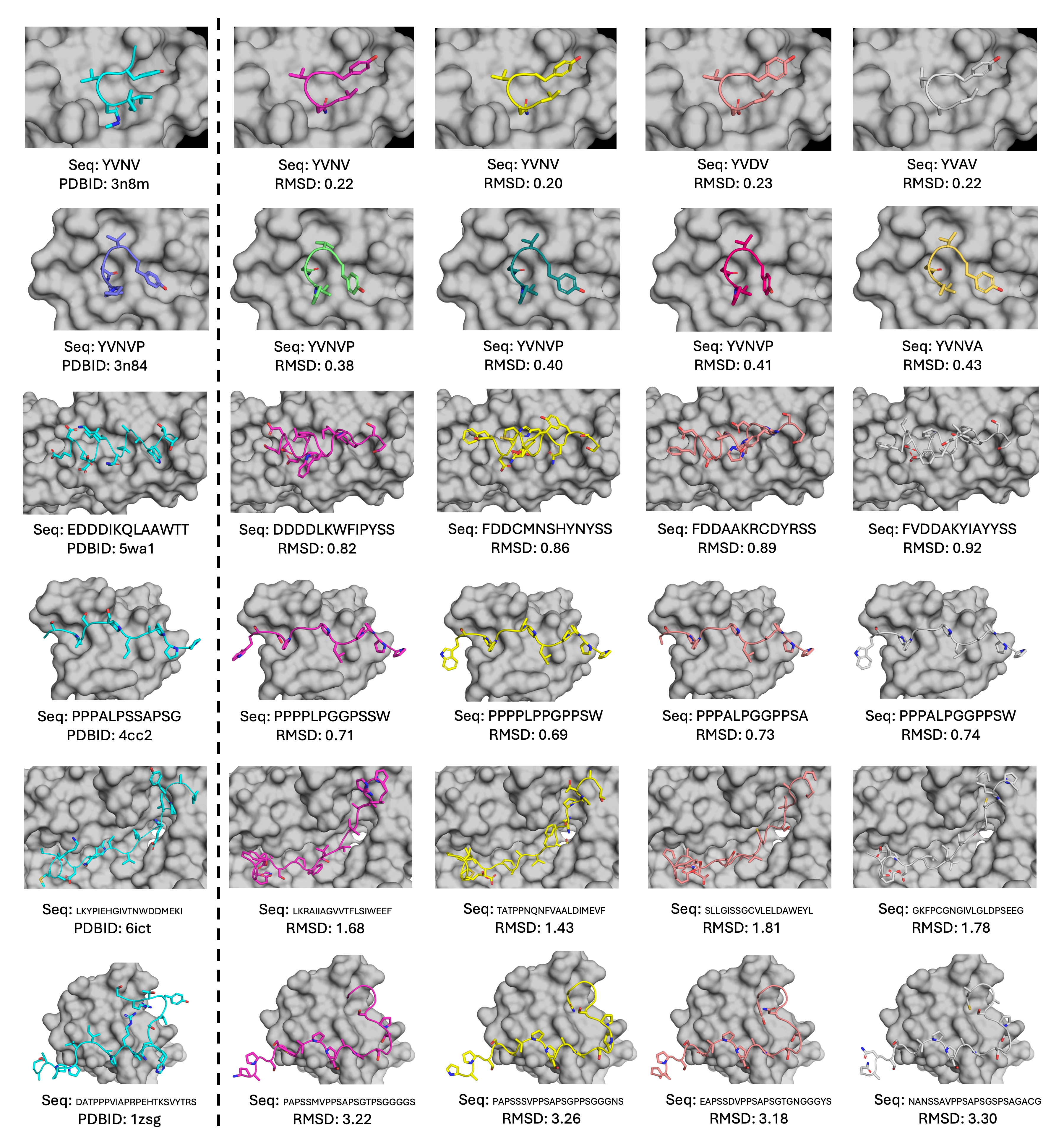}
    \caption{Visualization of peptides generated by POTFlow. The left column displays the lead peptide, whereas the right four columns showcase the peptides our method generated. We choose peptides of different lengths for our visualization. The first two rows have the shorter peptides, the middle two rows show peptides of a medium length, and the final two rows are made up of the longer peptides.}
    \label{app:vis}
\end{figure}

\section{ATP Concentrations in Different Cell Lines}
\begin{figure}[H]
    \centering
    \includegraphics[width=0.8\linewidth]{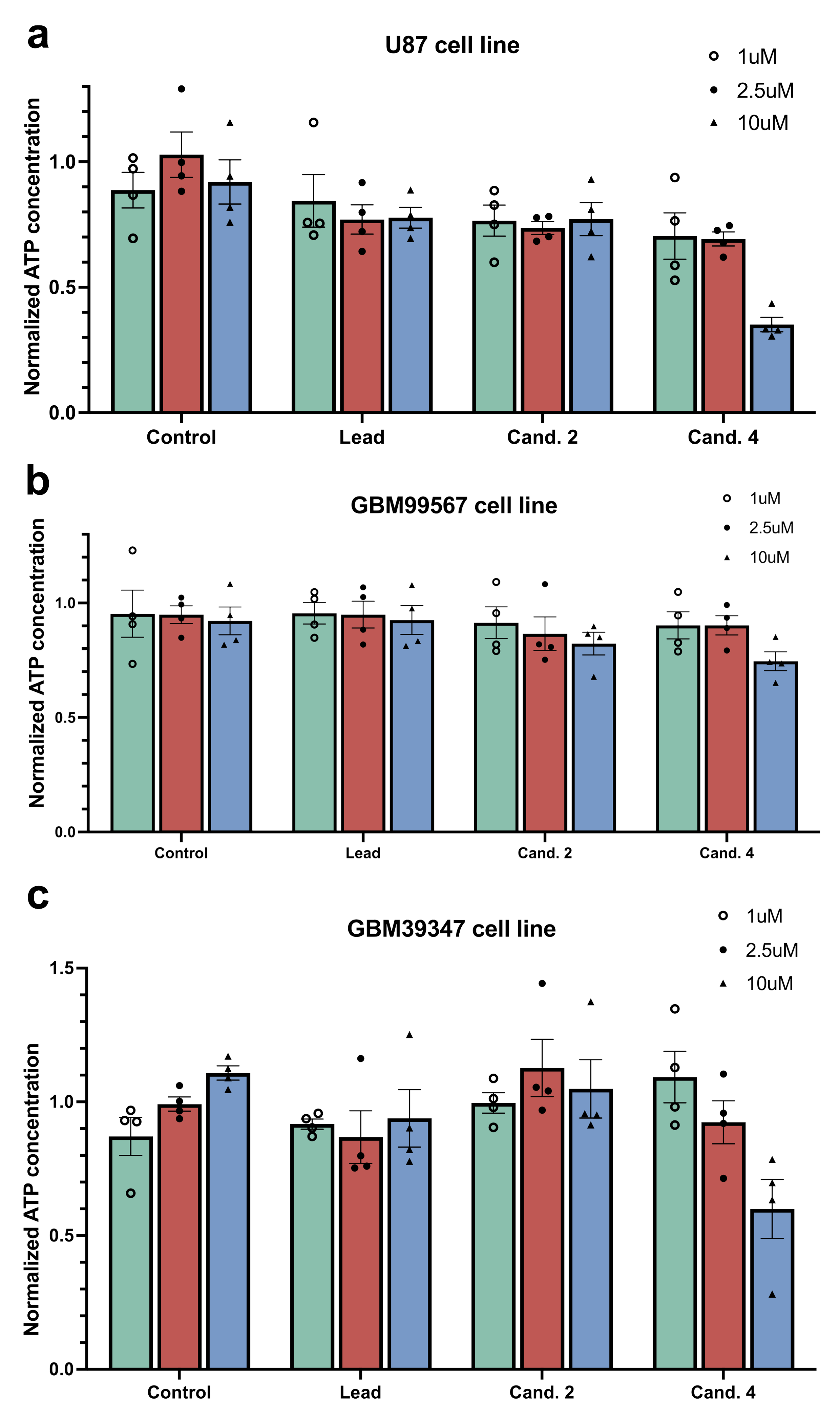}
    \caption{ATP concentrations in GBM cell lines (U87, 99567, and 39347) treated for 4 hours with control, the lead peptide, Cand.2, and Cand.4 at concentrations of 1.0$\mu$M (open circles), 2.5$\mu$M (filled circles), and 10~$\mu$M (triangles).}
    \label{fig_app:atp}
\end{figure}

\section{Inhibition of Viability Rate}
\begin{table*}[h]
    \centering
\renewcommand{\arraystretch}{1.5}
\resizebox{\linewidth}{!}{
    \begin{tabular}{cccccccccc}
        \toprule
        \multirow{2}{*}{Cell line} & \multicolumn{3}{c}{1$\mu$M} & \multicolumn{3}{c}{2.5$\mu$M} & \multicolumn{3}{c}{10$\mu$M} \\
        \cmidrule(r){2-4}
        \cmidrule(r){5-7}
        \cmidrule(l){8-10}
         & Lead & Cand.~4 & IVR \%  & Lead & Cand.~4 & IVR \%  & Lead & Cand.~4 & IVR \%     \\
        \midrule
        \multirow{2}{*}{Cell line 293} 
          & \multirow{2}{*}{0.60} & \multirow{2}{*}{0.54} & \multirow{2}{*}{10.00\%} & \multirow{2}{*}{0.57} & \multirow{2}{*}{0.50} & \multirow{2}{*}{12.28\%} & \multirow{2}{*}{0.56} & \multirow{2}{*}{0.52} & \multirow{2}{*}{7.14\%} \\
          (non-cancerous) &&&&&&&& \\
        \hline
        \multirow{2}{*}{Cell line 99567} 
          & \multirow{2}{*}{1.01} & \multirow{2}{*}{0.79} & \multirow{2}{*}{\textbf{21.78\%}} & \multirow{2}{*}{0.71} & \multirow{2}{*}{0.59} & \multirow{2}{*}{16.90\%} & \multirow{2}{*}{0.57} & \multirow{2}{*}{0.18} & \multirow{2}{*}{\textbf{68.42\%}} \\
          (GBM) &&&&&&&& \\
        \hline
        \multirow{2}{*}{Cell line U87} 
          & \multirow{2}{*}{0.55} & \multirow{2}{*}{0.45} & \multirow{2}{*}{\textbf{18.18\%}} & \multirow{2}{*}{0.54} & \multirow{2}{*}{0.46} & \multirow{2}{*}{14.81\%} & \multirow{2}{*}{0.53} & \multirow{2}{*}{0.40} & \multirow{2}{*}{\textbf{24.53\%}}\\
          (GBM) &&&&&&&& \\
        \bottomrule
    \end{tabular}
}   
    \vspace{1.5em}
    \caption{We summarize the cell viability measurements for the lead peptide and Cand.~4 across three cell lines (including one non-cancerous and two GBM cell lines) at three concentration levels (1, 2.5, and 10$\mu$M). The inhibition of viability rate (IVR\%) is computed to quantify the enhanced inhibitory effect of Cand.~4 over the lead peptide, defined as the relative reduction in viability. $ \text{IVR} = \frac{\text{Viability}_{\text{Lead}} - \text{Viability}_{\text{Cand.\ 4}}}{\text{Viability}_{\text{Lead}}} \times 100\% $.}
    \label{tab:cell_selectivity}
\end{table*}

\section{Log-rank Test in PDX Model}
\begin{table*}[h!]
\centering
\renewcommand{\arraystretch}{1.5}
\begin{tabular}{l c}
\hline
 Groups & p-value \\
\hline
Control vs Low dose   & 0.2325 \\
Low dose vs High dose & 0.4302 \\
Control vs High dose  & 0.0198$^*$ \\
\hline
\end{tabular}
\caption{Log-rank test results comparing survival curves between groups. $^*$ represent s statistically significant at the 0.05 level.}
\label{tab:Log-rank}
\end{table*}

\end{document}